\documentclass[sigplan,screen]{acmart}

\copyrightyear{2027}
\acmYear{2027}
\setcopyright{cc}
\setcctype{by}
\acmConference[ASPLOS '27]{Proceedings of the 32nd ACM International Conference on Architectural Support for Programming Languages and Operating Systems, Volume 1}{April 11--15, 2027}{Heraklion, Greece}
\acmBooktitle{Proceedings of the 32nd ACM International Conference on Architectural Support for Programming Languages and Operating Systems, Volume 1 (ASPLOS '27), April 11--15, 2027, Heraklion, Greece}
\acmDOI{10.1145/3845814.3850062}
\acmISBN{979-8-4007-3012-2/2027/04}

\usepackage{listings}
\usepackage{subcaption}
\usepackage{xspace}
\usepackage{wrapfig}
\usepackage{algorithm}
\usepackage{algpseudocode}
\usepackage{algorithmicx}
\usepackage{stmaryrd}
\usepackage{tikz}
\usepackage{dsfont}
\usetikzlibrary{positioning, calc, fit, backgrounds, automata, tikzmark, arrows, shapes, math, arrows.meta, patterns, fit, backgrounds}
\usepackage{cleveref}
\usepackage{url}

\begin{document}
\newcommand{\sys}{\textsc{Nacho}\xspace}
\newcommand{\taco}{\textsc{Taco}\xspace}

\definecolor{codebackground}{RGB}{252,252,252}

\lstdefinelanguage{cin}{
  keywords={forall, where, with},
  keywordstyle=\color{blue},
  identifierstyle=\color{black},
  sensitive=false,
  comment=[l]{//},
  morecomment=[s]{/*}{*/},
  commentstyle=\color{purple}\ttfamily,
  stringstyle=\color{red}\ttfamily,
  basicstyle=\scriptsize\ttfamily,
}

\lstdefinelanguage{cfir}{
  keywords={for, while, if, else, func, exclusive_scan, exclusive_prefix_sum, let},
  keywordstyle=\color{blue},
  identifierstyle=\color{black},
  sensitive=false,
  comment=[l]{//},
  morecomment=[s]{/*}{*/},
  commentstyle=\color{purple}\ttfamily,
  stringstyle=\color{red}\ttfamily,
  basicstyle=\scriptsize\ttfamily,
}

\lstdefinelanguage{cpp}{
  keywords={for, while, if, else, return, do, tid},
  keywordstyle=\color{blue},
  identifierstyle=\color{black},
  sensitive=false,
  comment=[l]{//},
  morecomment=[s]{/*}{*/},
  commentstyle=\color{purple}\ttfamily,
  stringstyle=\color{red}\ttfamily,
  basicstyle=\scriptsize\ttfamily,
  morekeywords = [2]{int, bool, float, void},
  keywordstyle = [2]\color{red},
  morekeywords = [3]{min, max, search, bsearch, lb_search},
  keywordstyle = [3]\color{teal},
  morekeywords = [4]{CSR, DCSR, Parts, SpVec},
  keywordstyle = [4]\color{purple},
}

\lstdefinelanguage{balance}{
  keywords={func, while, bsearch, if, else, elif, partition, return, for, in},
  keywordstyle=\color{blue},
  identifierstyle=\color{black},
  sensitive=false,
  comment=[l]{//},
  morecomment=[s]{/*}{*/},
  commentstyle=\color{purple}\ttfamily,
  stringstyle=\color{red}\ttfamily,
  basicstyle=\scriptsize\ttfamily,
  morekeywords = [2]{true, false, null},
  keywordstyle = [2]\color{red},
}

\newcommand{\aj}[1]{{\color{teal}(aj: #1)}}
\newcommand{\atharva}[1]{{\color{blue}(atharva: #1)}}
\newcommand{\fred}[1]{{\color{cyan}(fred: #1)}}
\newcommand{\rohan}[1]{{\color{olive}(rohan: #1)}}
\newcommand{\bobby}[1]{{\color{magenta}(bobby: #1)}}
\newcommand{\rubens}[1]{{\color{yellow}(rubens: #1)}}

\definecolor{todocolor}{rgb}{0.8,0,0}
\definecolor{editcolor}{rgb}{0,0,0.8}
\newcommand{\TODO}[1]{{\color{todocolor}(TODO: #1)}}
\newcommand{\TODON}[2]{{\color{todocolor}TODO (#1): #2}}
\newcommand{\IGNORE}[1]{}

\newcommand{\coordinate}[1]{c_{#1}}

\newcommand{\WorkOneD}[2]{\mathbf{W}_{#2, \, #1}\left(\coordinate{#1}\right)}
\newcommand{\WorkTwoD}[3]{\mathbf{W}_{#3, \, #2}\left( \coordinate{#1}, \, \coordinate{#2}\right)}
\newcommand{\WorkND}[3]{\mathbf{W}_{#3, \,#2}\left(\coordinate{#1}, \, \ldots, \, \coordinate{#2}\right)}
\newcommand{\WorkNDArg}[4]{\mathbf{W}_{#4, \,#2}\left(\coordinate{#1}, \, \ldots, \, #3\right)}

\newcommand{\WorkOneDArg}[3]{\mathbf{W}_{#3, \, #1}\left(#2\right)}
\newcommand{\WorkTwoDArg}[4]{\mathbf{W}_{#4, \, #1}\left(#2, \, #3\right)}

\newcommand{\WorkFunc}[2]{\mathbf{C}_{#1}\left(#2\right)}

\definecolor{color0}{RGB}{255,255,217}
\definecolor{color1}{RGB}{237,248,177}
\definecolor{color2}{RGB}{199,233,180}
\definecolor{color3}{RGB}{127,205,187}
\definecolor{color4}{RGB}{65,182,196}
\definecolor{color5}{RGB}{29,145,192}
\definecolor{color6}{RGB}{34,94,168}
\definecolor{color7}{RGB}{37,52,148}
\definecolor{color8}{RGB}{8,29,88}

\algrenewcommand\algorithmiccomment[1]{%
\hfill\textcolor{blue}{\(\triangleright\) #1}}
\algtext*{EndFor}
\algtext*{EndWhile}
\algtext*{EndIf}
\algtext*{EndFunction}
\algtext*{EndProcedure}

\title{Partitioning Unstructured Sparse Tensor Algebra for Load-Balanced Parallel Execution}

\author{Atharva Chougule}
\authornote{Equal contribution.}
\orcid{0009-0006-6369-6668}
\affiliation{%
  \institution{Stanford University}
  \streetaddress{353 Jane Stanford Way}
  \city{Stanford}
  \state{CA}
  \country{USA}
  \postcode{94305}
}
\email{atharvac@cs.stanford.edu}

\author{Alexander J Root}
\authornotemark[1]
\orcid{0000-0001-6221-1389}
\affiliation{%
  \institution{Stanford University}
  \streetaddress{353 Jane Stanford Way}
  \city{Stanford}
  \state{CA}
  \country{USA}
  \postcode{94305}
}
\email{ajroot@cs.stanford.edu}

\author{Rubens Lacouture}
\orcid{0009-0008-2268-0074}
\affiliation{%
  \institution{Stanford University}
  \streetaddress{353 Jane Stanford Way}
  \city{Stanford}
  \state{CA}
  \country{USA}
  \postcode{94305}
}
\email{rubensl@stanford.edu}

\author{Bobby Yan}
\orcid{0009-0002-6792-6222}
\affiliation{%
  \institution{Stanford University}
  \streetaddress{353 Jane Stanford Way}
  \city{Stanford}
  \state{CA}
  \country{USA}
  \postcode{94305}
}
\email{bobbyy@cs.stanford.edu}


\author{Rohan Yadav}
\orcid{0000-0003-0746-066X}
\affiliation{%
  \institution{Stanford University}
  \streetaddress{353 Jane Stanford Way}
  \city{Stanford}
  \state{CA}
  \country{USA}
  \postcode{94305}
}
\email{rohany@cs.stanford.edu}

\author{Fredrik Kjolstad}
\orcid{0000-0002-2267-903X}
\affiliation{%
  \institution{Stanford University}
  \streetaddress{353 Jane Stanford Way}
  \city{Stanford}
  \state{CA}
  \country{USA}
  \postcode{94305}
}
\email{kjolstad@cs.stanford.edu}

\renewcommand{\shortauthors}{Chougule and Root et al.}

\begin{abstract}

Sparse tensor algebra is challenging to efficiently parallelize due to
the irregular, data-dependent, and potentially skewed structure of sparse computation.
We propose the first partitioning algorithm that provably load balances the
computation of any sparse tensor algebra expression across parallel execution units.
Our algorithm generalizes parallel merging algorithms to any
number of operands, and to multi-dimensional, hierarchical sparse
data structures.
We implement our algorithm within an existing sparse tensor algebra compilation
framework to automatically generate parallel sparse tensor algebra kernels
that target multi-core CPUs and GPUs.
We show that our generated code is competitive with hand-implemented
parallelization strategies used by vendor libraries like
Intel MKL and NVIDIA cuSPARSE (geo-means of $0.93$--$5.5\times$) and \taco (geo-means of $1.2$--$2.4\times$), and
significantly outperforms PyTorch Sparse's general-purpose strategies
on sparse tensor expressions where specialized algorithms
have not been developed (geo-means of $2.0$--$9.1\times$).
%


\end{abstract}

\ccsdesc[500]{Software and its engineering~Source code generation}
\ccsdesc[500]{Computing methodologies~Parallel algorithms}

\keywords{sparse tensor algebra, parallelism, load balance, GPUs, multicore CPUs, partitioning, compilers}

\maketitle 

\section{Introduction}
\label{sec:introduction}

Sparse tensor algebra is used across many computational domains~\cite{frankle2019lottery, kolecki2002introduction, kats2013stf}, and high-performance implementations of sparse tensor algebra operations
enable these domains to solve important problems quickly.
Parallelization is a critical component of efficient execution on modern machines, and the main
challenge in parallelizing sparse tensor algebra lies in load balancing the irregular work
created from \textit{coiteration}.
%
%
%
%
Coiteration is the process of simultaneously iterating over multiple sparse tensors 
and performing different operations based on how the coordinates in each sparse data 
structure match.
The total amount of work performed by coiteration is data-dependent and a function of the interaction between
sparse data structures, making extracting equally-sized pieces of work a search problem.
This difficulty is exacerbated by the increased irregularity that comes from the hierarchical structure in both
storing and coiterating over multi-dimensional sparse tensors.

The community has developed a variety of compilers that take arbitrary sparse tensor expressions and descriptions of the underlying data structures and produce efficient implementations~\cite{kjolstad2017taco, ahrens2023looplets, ahrens2025finch, shaikhha2022sdql}.
Extensions of these compilers support various types of parallel hardware, e.g., multicore CPUs~\cite{wilkinson2023regtiling, yan2026scorch, ghorbani2025compressed, zhao2022spf}, GPUs~\cite{senanayake2020scheduling, won2026insum, ye2023sparsetir, bansal2023mosaic, zhang2016hybridsparse}, supercomputers~\cite{yadav2022spdistal}, or custom accelerators~\cite{hsu2025stardust, koul2025onyx, siracusa2026ember}.
However, these parallelization strategies either do not guarantee the creation of equally-sized pieces of work for each parallel worker~\cite{senanayake2020scheduling, yadav2022spdistal, hsu2025stardust, bansal2023mosaic, koul2025onyx, yan2026scorch, ghorbani2025compressed, zhao2022spf} or can only do so for expressions with a single sparse tensor~\cite{ye2023sparsetir, won2026insum, wilkinson2023regtiling, siracusa2026ember, zhang2016hybridsparse}.
Additionally, these compilers often require user input in the form of scheduling directives to
choose a parallelization strategy~\cite{senanayake2020scheduling, bansal2023mosaic, yadav2022spdistal, siracusa2026ember, hsu2025stardust, koul2025onyx, ye2023sparsetir}.

For specific kernels and specific sparse data structures, the community has developed algorithms that 
guarantee an equal distribution of work among parallel processors~\cite{dalton2015mergepath, yang2015spmspv, yang2022graphblast, xiao2023survey, odemuyiwa2023drt, jain2025rassm}.
Many of these algorithms~\cite{dalton2015mergepath, yang2015spmspv, yang2022graphblast, xiao2023survey} are inspired by parallel merging algorithms from the sorting
community~\cite{saher2012mergepath, green2012gpumergepath}, and extract load-balanced
parallelism from the merge-like operations that implement coiteration in sparse
tensor algebra kernels.
These strategies are specialized for specific binary operations and only handle one and two-dimensional tensors.
Other strategies developed for custom accelerators~\cite{odemuyiwa2023drt, jain2025rassm} collect statistics across the operands to adapt partitions, but do so in a sequential manner.
%

We propose a load-balanced parallel partitioning technique that guarantees equally-sized pieces of work
(within a small constant) for any number of tensors of any number of dimensions, and generalizes the ideas proposed in prior algorithms for specific expressions and data structures.
Our partitioning technique efficiently searches the $n$-dimensional iteration space of the tensor algebra expression to find a load-balanced partitioning with respect to an abstract cost model.
%
The key insight behind our partitioning strategy is that load-balanced \textit{irregular} partitions of the iteration space can be extracted from a combination of the loop structure and efficient queries on sparse data structures.
Critically, these queries can be performed independently, allowing for the partitioning strategy
itself to be computed \emph{in parallel} rather than as a separate pre-processing step.
We implement a concrete cost model that corresponds to the number of non-zero values within
a partition, which can be queried efficiently for many sparse tensor formats.
%
By specializing our partitioning technique with this format-specific cost model and a concrete loop order, we generate custom partitioning algorithms for sparse kernels within a sparse tensor algebra compiler.

We implement our partitioning approach in a prototype compiler called \sys{},\footnote{The ideal plate of nachos load-balances cheese across all chips.}
which extends the \taco{}~\cite{kjolstad2017taco, chou2018formats} compiler to generate load-balanced parallel kernels for a general class of expressions and sparse data structures. 
To the best of our knowledge, \sys{} is the first compiler to support load-balanced GPU code generation of 
sparse tensor algebra kernels with multiple sparse operands.
Unlike prior approaches, our approach is fully automated and does not require user input (in the form of scheduling
directives) to generate parallel code.
The contributions of our work are: 
\begin{itemize}
    \item A parallel partitioning strategy for sparse tensor algebra that guarantees load-balanced parallel execution.
    \item A code generation approach that integrates this partitioning strategy into a sparse compilation framework.
\end{itemize}
To evaluate our contributions, we compare the performance of \sys{} on a variety
of sparse tensor algebra expressions on multi-core CPUs and GPUs.
We show that \sys{} achieves
high performance when compared to tuned vendor libraries like Intel MKL~\cite{mkl}, NVIDIA cuSPARSE~\cite{cusparse}, and the \taco compiler~\cite{kjolstad2017taco, senanayake2020scheduling}. 
We then show that the generality of our approach enables high performance 
for expressions not present in vendor libraries, heavily outperforming generic strategies
in existing systems like PyTorch Sparse~\cite{pytorch_sparse}.
\section{Motivating Example and Challenges}
\label{sec:motivation}
We now use several examples to illustrate the concrete challenges of
parallelizing multi-sparse tensor algebra operations.
We start with coiteration over a single dimension before moving into 
additional challenges caused by multi-dimensional tensors as well as nested intersection iteration patterns.


\subsection{Background: Coiteration}
\label{sec:coiteration-background}


\begin{figure}
  
  \centering
\begin{lstlisting}[language=cpp, mathescape=true, backgroundcolor=\color{codebackground}, frame=single,
caption={Sequential code computing $z_{i} = a_{i} \cdot b_{i} \cdot c_{i}$ via a three-finger merge on sparse vectors. $\texttt{nnz}_x$ denotes the count of non-zeros in vector $x$ and $p_x$ is the iterator over vector $x$.},
label={fig:vec3-compute}]
$p_a$ = 0, $p_b$ = 0, $p_c$ = 0, $p_z$ = 0
while $p_a$ < nnz($a$) && $p_b$ < nnz($b$) && $p_c$ < nnz($c$):
  $i_a$ = a.i[$p_a$], $i_b$ = b.i[$p_b$], $i_c$ = c.i[$p_c$]
  i = min($i_a$, $i_b$, $i_c$)
  if i == $i_a$ && i == $i_b$ && i == $i_c$:
    z.v[$p_z$] = a.v[$p_a$] * b.v[$p_b$] * c.v[$p_c$]
    z.i[$p_z$++] = i
  $p_a$ += (i${==}i_a$); $p_b$ += (i${==}i_b$); $p_c$ += (i${==}i_c$)
\end{lstlisting}
\vspace{-1.25em}
\end{figure}

\begin{figure}
  \centering
  \begin{tikzpicture}
    \pgfmathsetmacro{\cs}{0.5}   
    \pgfmathsetmacro{\gp}{0.3}    

    \newcommand{\ivvec}[4]{%
      \begin{scope}[xshift=#2 cm]
        \foreach \dummy [count=\k from 1] in {#3} { \xdef\Niv{\k} }
        \pgfmathsetmacro{\labelx}{\Niv*\cs*0.5}
        \node[font=\small] at (\labelx, 2*\cs+0.2) {$#1$};
        \foreach \x [count=\col from 0] in {#3} {
          \draw (\col*\cs, \cs) rectangle (\col*\cs+\cs, 2*\cs);
          \node[font=\scriptsize] at (\col*\cs+0.5*\cs, 1.5*\cs) {$\x$};
        }
        \foreach \x [count=\col from 0] in {#4} {
          \draw (\col*\cs, 0) rectangle (\col*\cs+\cs, \cs);
          \node[font=\scriptsize] at (\col*\cs+0.5*\cs, 0.5*\cs) {$\x$};
        }
      \end{scope}
    }

    \node[anchor=east, font=\small] at (-0.1, 1.5*\cs) {$.i$};
    \node[anchor=east, font=\small] at (-0.1, 0.5*\cs) {$.v$};

    \pgfmathsetmacro{\xA}{0}
    \pgfmathsetmacro{\xB}{3*\cs + \gp}
    \pgfmathsetmacro{\xC}{8*\cs + 2*\gp}

    \ivvec{a}{\xA}{1, 4, 9}         {a_0, a_1, a_2}
    \ivvec{b}{\xB}{1, 2, 4, 9, 10}  {b_0, b_1, b_2, b_3, b_4}
    \ivvec{c}{\xC}{1, 5, 9, 11}    {c_0, c_1, c_2, c_3}
  \end{tikzpicture}
  \caption{Sparse vectors $a$, $b$, and $c$, as index/value pairs.}
  \label{fig:sparse-vec-rep}
  \vspace{-1.5em}
\end{figure}

Sparse tensor compilers~\cite{kjolstad2017taco, strout2018spf, ahrens2023looplets, ahrens2025finch, bik2022mlir, shaikhha2022sdql} take advantage of the arithmetic properties $a \times 0 = 0$ and $a + 0 = a$ to iterate over only the \textit{intersection} of non-zero coordinates to compute the result of a sparse multiplication, and the \textit{union} of non-zero coordinates to compute the result of a sparse addition.
There are multiple ways to perform such iteration, including the iterate-locate pattern for intersections~\cite{kjolstad2017taco, root2024burrito, zhao2022spf}, which iterates one operand and searches for a matching coordinate in the other operand.
While this is asymptotically optimal if the searched operand is stored densely (and therefore stores values for all coordinates), it can be asymptotically sub-optimal if both operands are sparse.
Iterate-locate also does not naively apply to unions of sparse operands.
Therefore, \sys follows \taco's approach for kernels with multiple sparse operands: generating \textit{coiterative} while loops, similar to the merge operation of mergesort, to step through each sparse operand.

An example of a multi-way merge is provided in ~\Cref{fig:vec3-compute}, which computes the element-wise multiplication of three sparse vectors via a three-way merge.
This multi-way merge maintains a pointer into each sparse vector and increments the pointers iteratively to compute the intersection, and performs the element-wise multiplication 
when all pointers match on a non-zero coordinate.
This coiterative {\tt while} loop is not amenable to standard data parallelization techniques; but can still be parallelized via a technique from the parallel sorting community.

\subsection{Parallelizing Coiteration}

The parallel sorting community developed a technique to parallelize a coiterative \textit{merge} of two vectors in a load-balanced manner~\cite{saher2012mergepath, green2012gpumergepath}.
This technique searches the two sorted arrays to find non-overlapping index ranges that can be safely merged independently.
%

Our approach in \sys adapts and generalizes this strategy to parallelize the coiterative intersections and unions
performed by sparse tensor algebra kernels.
\sys{} partitions the dense space of all possible coordinates
in each vector (referred to as the \emph{coordinate space}) into \emph{irregular} pieces,
such that the number of non-zero elements processed in each partition is equal,
rather than the total number of coordinates.
\sys{} finds these partition boundaries by binary searching over the coordinate
space for ranges that enclose an equal number of non-zero coordinates across
each sparse vector.
When considering a candidate range $[x_l, x_h)$, the number of non-zero coordinates 
of each sparse vector contained within $[x_l, x_h)$ can be found by another binary search
over the physical data structures storing the sparse vectors.

\begin{figure}[t]
    \centering
    \newcommand{\setupcolors}[1]{
        \foreach \bound/\col [remember=\bound as \rawprevbound (initially 0)] in {#1} {
            
            \pgfmathtruncatemacro{\startidx}{\rawprevbound}
            \pgfmathtruncatemacro{\endidx}{\bound-1}
            
            \ifnum\endidx<\startidx\relax
            \else
                \foreach \i in {\startidx,...,\endidx} {
                    \global\expandafter\edef\csname cellcolor\i\endcsname{\col}
                }
            \fi
        }
    }

    \newcommand{\drawdomain}[3][12]{
        \begin{scope}[shift={(0, #2 * \offset)}]
            \pgfmathtruncatemacro{\maxcol}{#1-1}
            
            \foreach \x in {0,...,\maxcol} {
                \edef\mycolor{\ifcsname cellcolor\x\endcsname\csname cellcolor\x\endcsname\else white\fi}
                \fill[\mycolor!80] (\x, 0) rectangle (\x+1, 1);
            }
            
            \draw (0, 0) grid (#1, 1);
            \node[anchor=east] at (-0.5, 0.5) {#3};
            
            \foreach \x in {0,...,\maxcol} {
                \node at (\x+0.5, 0.5) {\x};
            }
        \end{scope}
    }

    \newcommand{\drawvector}[4][12]{
        \begin{scope}[shift={(0, #2 * \offset)}]
            
            \foreach \idx/\val in {#4} {
                \edef\mycolor{\ifcsname cellcolor\idx\endcsname\csname cellcolor\idx\endcsname\else white\fi}
                \fill[\mycolor!80] (\idx, 0) rectangle (\idx+1, 1);
            }
            
            \draw (0, 0) grid (#1, 1);
            \node[anchor=east] at (-0.5, 0.5) {#3};
            
            \foreach \idx/\val in {#4} {
                \node at (\idx+0.5, 0.5) {\val};
            }
        \end{scope}
    }

    \newcommand{\drawpartitionlines}[2]{
        \pgfmathsetmacro{\bottomedge}{#2 * \offset - 0.2}
        \foreach \x in {#1} {
            \draw[line width=1.5pt] (\x, 1.2) -- (\x, \bottomedge);
        }
    }

    \newcommand{\drawpartitionlabels}{
        \pgfmathsetmacro{\labely}{3 * \offset - 0.7}
        \node at (1, 1.7) {$[0,2)$};
        \node at (3.5, 1.7) {$[2,5)$};
        \node at (7.5, 1.7) {$[5,10)$};
        \node at (11, 1.7) {$[10,12)$};
    }

    \newcommand{\drawboundarylabels}{
        \pgfmathsetmacro{\labely}{3 * \offset - 0.7}
        \node at (1, \labely) {$p=0$};
        \node at (3.5, \labely) {$p=1$};
        \node at (7.5, \labely) {$p=2$};
        \node at (11, \labely) {$p=3$};
    }

    \begin{tikzpicture}[scale=0.4]
        \def\offset{-1.5} 

        \setupcolors{2/color1,5/color3,10/color5,12/color7}

        \drawdomain[12]{0}{Coordinate space}
        
        \drawvector{1}{Sparse vector $a$}{1/$a_0$, 4/$a_1$, 5/$a_2$}
        \drawvector{2}{Sparse vector $b$}{1/$b_0$, 2/$b_1$, 4/$b_2$, 9/$b_3$, 10/$b_4$}
        \drawvector{3}{Sparse vector $c$}{1/$c_0$, 5/$c_1$, 10/$c_2$, 11/$c_3$}

        \drawpartitionlines{0,2,5,10,12}{3}

        \drawpartitionlabels

        \drawboundarylabels
        
    \end{tikzpicture}
    \vspace{-.5em}
    \caption{An irregular partitioning of the coordinate space into 4 parts such that 3 non-zeros are processed in each partition. Partition $p$'s  boundaries $[x_{l}, x_{h})$ are obtained by independently searching for $x_l$ and $x_{h}$ s.t. $[0, x_l)$ contains $3p$ non-zeros and $[0, x_{h})$ contains $3(p+1)$ non-zeros.}
    \label{fig:vec_lb_example}
    \vspace{-1.5em}
\end{figure}

\Cref{fig:vec_lb_example} illustrates a partitioning found
by using this search for three-way merge.
The partition boundaries are found by searching for cutoff points in the coordinate
space that enclose a specific number of non-zeros.
For example, the first partition is found by searching for cutoff coordinate $x_1$ such that the range of coordinates
from $[0, x_1)$ contains $k$ non-zeros.
Then, the second partition $[x_1, x_2)$ can be computed by
also searching for $x_1$ and then searching for $x_2$ where the range of coordinates
$[0, x_2)$ contains $2k$ non-zeros.

%

Critically, \sys{}'s strategy allows for the bounds of each range in the partition
to be computed independently, and thus \emph{in parallel}.
%
As a result, the strategy can be applied to~\Cref{fig:vec3-compute} by having parallel workers
independently compute the partition bounds shown in~\Cref{fig:vec_lb_example}, 
and then begin load-balanced coiteration within the corresponding ranges.

\subsection{Higher-Order Coiteration}

While the three-way sparse vector multiplication example shown previously
builds intuition for parallelizing coiteration of multiple sparse operands,
generalizing these ideas also requires extending to multiple tensor dimensions.
Consider the coiteration required for an addition between 
two Doubly Compressed Sparse Row (DCSR)~\cite{buluc2008dcsr} matrices
in~\Cref{lst:eadd-dcsr2-cfir}.
This hierarchical coiteration requires a partitioning strategy that accounts for work across multiple dimensions.
\Cref{fig:dcsr-add-partition} depicts one such strategy, with a two-dimensional irregular partitioning to load balance the number of non-zeros processed in each partition.

\begin{figure}
\begin{lstlisting}[language=cfir, mathescape=true, basicstyle=\footnotesize\ttfamily, backgroundcolor=\color{codebackground}, frame=single,
caption={$Z_{ij} = A_{ij}+B_{ij}$ with DCSR inputs.},
label={lst:eadd-dcsr2-cfir}]
while $i~\shortleftarrow$ rows($A$) $\cup$ rows($B$):
  while $j~\shortleftarrow$ cols($A[i]$) $\cup$ cols($B[i]$):
    $Z[i, j]$ = $A[i, j]$ + $B[i, j]$
\end{lstlisting}

\begin{lstlisting}[language=cfir, mathescape=true, basicstyle=\footnotesize\ttfamily, backgroundcolor=\color{codebackground}, frame=single,
caption={$Z_{ij} = A_{ij}\cdot B_{ij}$ with DCSR inputs.},
label={lst:emul-dcsr2-cfir}]
while $i~\shortleftarrow$ rows($A$) $\cap$ rows($B$):
  while $j~\shortleftarrow$ cols($A[i]$) $\cap$ cols($B[i]$):
    $Z[i, j]$ = $A[i, j]$ * $B[i, j]$
\end{lstlisting}

\end{figure}

The partitioning problem is further complicated
by \emph{hierarchical skipping}, as observed in the Hadamard product on DCSR matrices
in~\Cref{lst:emul-dcsr2-cfir}.
%
Here, coiteration over the $j$ dimension is restricted to rows that are nonempty in both operands.
The sparse intersection in the outer loop causes entire rows to be skipped, and consequently the inner iteration over columns is skipped for those rows.
This makes computing load-balanced partitions difficult, as large regions of the coordinate space may correspond to no non-zeros.
Achieving load-balance in this case requires determining which rows contribute to the work and then constructing partitions only over these rows, as visualized in~\Cref{fig:dcsr-mul-partition}.

\begin{figure}[b]
\centering
\newcommand{\drawsparsematrix}[6]{
\begin{scope}[yscale=-1, xscale=1]

        \edef\skiplist{#5}%
        \foreach \r in \skiplist {
            \fill[orange!10] (0, \r) rectangle (#2, \r+1);
        }

        \edef\boundarylist{#4}%
        \edef\coordlist{#3}%

        \edef\colorlist{#6}
        \foreach \r/\c [count=\i] in \coordlist {
            \foreach \col [count=\j] in \colorlist {
                \ifnum\i=\j
                    \fill[\col] (\c, \r) rectangle (\c+1, \r+1);
                \fi
            }
        }
        
        \draw[black!40, thin] (0, 0) grid (#2, #1);
        
        \edef\boundarylist{#4}%
        \foreach \xone/\yone/\xtwo/\ytwo in \boundarylist {
            \draw[black, line width=1.75pt] (\xone, \yone) -- (\xtwo, \ytwo);
        }

        \foreach \r in \skiplist {
            \draw[orange, thick, dashed] (0, \r+0.5) -- (#2, \r+0.5);
        }

    \end{scope}
}

\newcommand{\matrixA}{
    0/0, 0/3, 0/4,
    0/5, 2/1, 2/4,
    2/6, 2/7,
    4/0, 4/1, 4/4, 4/6}

\newcommand{\matrixAadd}{
    color1, color1, color1,
    color3, color3, color3,
    color5, color5,
    color7, color7, color7, color7}

\newcommand{\matrixAmul}{
    color1, color1, color3,
    color3, color5, color5,
    color7, color7,
    gray, gray, gray, gray}

\newcommand{\matrixB}{
    0/0, 0/3, 0/4,
    2/0, 2/2, 2/3, 2/6, 2/7,
    3/1, 3/3, 3/4, 3/7}

\newcommand{\matrixBadd}{
    color1, color1, color1,
    color3, color3, color3,
    color5, color5, color5, color5,
    color7, color7}

\newcommand{\matrixBmul}{
    color1, color1, color3,
    color3, color5, color5,
    color7, color7,
    gray, gray, gray, gray}

\newcommand{\addpartitionlines}{
    0/0/0/1,
    5/0/5/1,
    5/2/5/3,
    4/3/4/4,
    8/0/8/2,
    8/3/8/5,
    0/1/0/3,
    0/4/0/5,
    0/0/8/0,
    0/1/5/1,
    5/2/8/2,
    0/3/8/3,
    0/4/4/4,
    0/5/8/5
}

\newcommand{\mulpartitionlines}{
    0/0/0/1,
    4/0/4/1,
    1/2/1/3,
    5/2/5/3,
    8/2/8/3,
    0/0/8/0,
    0/1/8/1,
    0/2/8/2,
    0/3/8/3
}

\newcommand{\mulskiplines}{1, 3, 4}

\newcommand{\addpartcolors}{color1/0/0/4/0, color3/5/0/4/2, color5/5/2/3/3, color7/4/3/7/4}

\begin{subfigure}{0.48\textwidth}
\centering
\begin{tikzpicture}[scale=0.35, baseline=(current bounding box.center)]
    \drawsparsematrix{5}{8}
        {\matrixA}
        {\addpartitionlines}
        {} 
        {\matrixAadd}
\end{tikzpicture}
\hspace{2em}$\vcenter{\hbox{\Large $\oplus$}}$\hspace{2em}
\begin{tikzpicture}[scale=0.35, baseline=(current bounding box.center)]
    \drawsparsematrix{5}{8}
        {\matrixB}
        {\addpartitionlines}
        {} 
        {\matrixBadd}
\end{tikzpicture}
\caption{Partitioning for DCSR addition: 6 non-zeros per chunk.}
\label{fig:dcsr-add-partition}
\end{subfigure}
\vfill
\vspace{1em}
\begin{subfigure}{0.48\textwidth}
\centering
\begin{tikzpicture}[scale=0.35, baseline=(current bounding box.center)]
    \drawsparsematrix{5}{8}
        {\matrixA}
        {\mulpartitionlines}
        {\mulskiplines}
        {\matrixAmul}
\end{tikzpicture}
\hspace{2em}$\vcenter{\hbox{\Large $\odot$}}$\hspace{2em}
\begin{tikzpicture}[scale=0.35, baseline=(current bounding box.center)]
    \drawsparsematrix{5}{8}
        {\matrixB}
        {\mulpartitionlines}
        {\mulskiplines}
        {\matrixBmul}
\end{tikzpicture}
\caption{Partitioning for DCSR (element-wise) multiplication. Entire rows are skipped (highlighted in orange) due to hierarchical skipping.}
\label{fig:dcsr-mul-partition}
\end{subfigure}
\caption{Two partitioning strategies that load-balance the number of non-zeros iterated across 4 parallel threads. Colors denote partition membership, where gray non-zeros are not part of any partition. Due to the hierarchical skipping in \textbf{(b)}, it is impossible to load balance iterated non-zeros without knowing which rows will be skipped versus iterated.}
\label{fig:dcsr-partitions}
\vspace{-1em}
\end{figure}
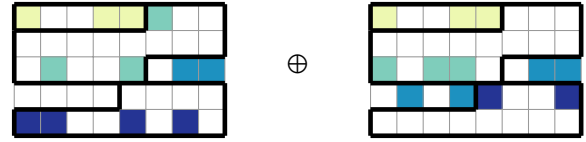
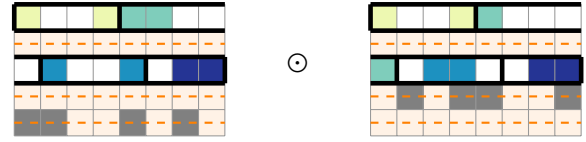

\section{Partitioning Algorithm}
\label{sec:partition}


We now discuss a generalized framework and search procedure for computing partitions that scales to
sparse tensor algebra expressions that contain any number of tensors,
any number of dimensions, and any combination of unions and intersections.
%
%
This search procedure is defined with respect to abstract cost
functions that model the amount of work performed within each partition;
the definitions of these functions are specialized to the concrete
tensor algebra expression and sparse data structure, and are described
in~\Cref{sec:codegen}.
The produced partitions are non-overlapping ranges of coordinates that can be processed in parallel: this requires following a lexicographic ordering, the ordering of loop nest variables.
We first describe the class of cost models supported by our search procedure in 
\Cref{sec:work-funcs}.
Then, in~\Cref{sec:hierarchical-search}, we describe a hierarchical search strategy that can 
compute partitions for tensor algebra expressions with any number of dimensions and 
union iteration patterns between operands.
Finally, in~\Cref{sec:recursive-partitioning}, we introduce a recursive variant 
of the search that progressively computes partitions for nested intersection iteration patterns,
which require a different strategy due to hierarchical skipping.

\subsection{Cost Functions}
\label{sec:work-funcs}

A candidate sparse tensor algebra expression is defined over
a $d$-dimensional iteration space $\mathcal{I}$.
We assume these dimensions are ordered from 0 to $d-1$ by 
some loop ordering $\mathcal{L}$.
We define a class of cost functions $\mathbf{C}_m$ for each
dimension $m$ in $\mathcal{I}$.
Let $\mathbf{N}_m$ refer to the extent of dimension $m$ in
$\mathcal{I}$.
Then, $\mathbf{C}_m (x \mid \, x_0, \ldots, x_{m-1})$ models the cost of
coiterating lexicographically with respect to $\mathcal{L}$ between
the $d$-dimensional coordinates
$(x_0, \ldots, x_{m-1}, 0, 0, \ldots, 0)$ and
$(x_0, \ldots, x_{m-1}, x-1,\allowbreak  \mathbf{N}_{m+1}, \allowbreak  \ldots, \mathbf{N}_{d-1})$.
Intuitively, this subspace of $\mathcal{I}$ fixes coordinates
for the $0$ to $m-1$ dimensions of $\mathcal{I}$, limits
the space of coordinates for dimension $m$ to be within $[0, x)$ and
spans the full extent of all remaining dimensions $m+1$ to $d-1$.

Various subspaces defined by $\mathbf{C}_m$ are shown in~\Cref{fig:coordinate-tree}.
\begin{figure}[t]
\centering
\newcommand{\registernode}[1]{%
  \expandafter\def\csname present@#1\endcsname{1}%
}

\newcommand{\nodewithlabel}[6]{%
  \if1#1%
    \node[
      draw, circle,
      minimum size=4.5mm,
      inner sep=0.5pt,
      font=\scriptsize,
      fill=white
    ] (#6) at (#2,#3) {$#4_{#5}$};%
  \else
    \node[
      draw, circle, dotted,
      minimum size=4.5mm,
      inner sep=0.5pt,
      font=\scriptsize,
      fill=white,
      text=black
    ] (#6) at (#2,#3) {$#4_{#5}$};%
  \fi
}

\newcommand{\drawcoordinatetree}[9]{%
  \begin{tikzpicture}[scale=0.8, transform shape]
    \begingroup
    \fill[color2, rounded corners=1mm]
    (-5.4,-1.9) -- (-5.4,0.3) -- (-1.82,0.3) -- (-1.82,-1.9) -- cycle;

  \fill[color4, rounded corners=1mm]
    (-1.78,-1.9) -- (-1.78,-0.5) -- (0.58,-0.5) -- (0.58,-1.9) -- cycle;

  \fill[color6, rounded corners=1mm]
    (0.6,-1.9) -- (0.6,-1.3) -- (1.2,-1.3) -- (1.2,-1.9) -- cycle;
  \node[anchor=west, font=\small, draw, rectangle] (#4) at (-0.225, 0.7) {$#4$};
  \forcsvlist{\registernode}{#9}

  \foreach \ii in {0,...,#1} {
    \pgfmathsetmacro{\xI}{(\ii-1)*#5}

    \xdef\Ipresent{0}
    \foreach \jj in {0,...,#2} {
      \foreach \kk in {0,...,#3} {
        \ifcsname present@\ii/\jj/\kk\endcsname
          \xdef\Ipresent{1}
        \fi
      }
    }

    \nodewithlabel{\Ipresent}{\xI}{0}{i}{\ii}{inode\ii}
    \if1\Ipresent
            \draw[ thin, black] (#4) -- (inode\ii);
      \else
        \draw[thin,black, dotted] (#4) -- (inode\ii);
        \fi
    \foreach \jj in {0,...,#2} {
      \pgfmathsetmacro{\xJ}{\xI + (\jj - (#2)/2.0)*#6}
      \pgfmathsetmacro{\yJ}{-#8}

      \xdef\Jpresent{0}
      \foreach \kk in {0,...,#3} {
        \ifcsname present@\ii/\jj/\kk\endcsname
          \xdef\Jpresent{1}
        \fi
      }


      \nodewithlabel{\Jpresent}{\xJ}{\yJ}{j}{\jj}{jnode\ii-\jj}

      \if1\Ipresent
          \if1\Jpresent
            \draw[thin, black] (inode\ii) -- (jnode\ii-\jj);
          \else
            \draw[thin,black, dotted] (inode\ii) -- (jnode\ii-\jj);
          \fi
      \else
        \draw[thin,black, dotted] (inode\ii) -- (jnode\ii-\jj);
        \fi

      \foreach \kk in {0,...,#3} {
        \pgfmathsetmacro{\xK}{\xJ + (\kk - (#3)/2.0)*#7}
        \pgfmathsetmacro{\yK}{-2*#8}

        \xdef\Kpresent{0}
        \ifcsname present@\ii/\jj/\kk\endcsname
          \xdef\Kpresent{1}
        \fi


        \nodewithlabel{\Kpresent}{\xK}{\yK}{k}{\kk}{knode\ii-\jj-\kk}

        \if1\Jpresent
            \if1\Kpresent
              \draw[thin, black] (jnode\ii-\jj) -- (knode\ii-\jj-\kk);
            \else
              \draw[thin, black, dotted] (jnode\ii-\jj) -- (knode\ii-\jj-\kk);
            \fi
        \else
              \draw[thin, black, dotted] (jnode\ii-\jj) -- (knode\ii-\jj-\kk);
        \fi
      }
    }
  }
    \endgroup
  \end{tikzpicture}%
}

\newcommand{\drawtwotreesB}[9]{%
  \drawcoordinatetree
    {#1}
    {#2}
    {#3}
    {#4}
    {#5}
    {#6}
    {#7}
    {#8}
    {#9}
}
\newcommand{\drawtwotrees}[9]{%
  \drawcoordinatetree
    {#1}
    {#2}
    {#3}
    {#4}
    {#5}
    {#6}
    {#7}
    {#8}
    {#9}
  \par\vspace{3pt}%
  \drawtwotreesB%
}

\drawtwotrees
  {2}
  {2}
  {1}
  {A}
  {3.6}
  {1.2}
  {0.65}
  {0.8}
  {0/0/0, 0/2/1, 0/2/0, 1/0/1, 1/2/0, 1/2/1}
  {2}
  {2}
  {1}
  {B}
  {3.6}
  {1.2}
  {0.65}
  {0.8}
  {1/0/0, 1/1/0, 1/2/0, 1/1/1, 2/0/0, 2/0/1, 2/1/1}
\par\vspace{4pt}%
 \begin{tikzpicture}[scale=0.8, transform shape]
     \fill[color2] (0,0) rectangle (0.4,0.35);
    \node[anchor=west, font=\small] at (0.5, 0.125) {$\mathbf{C}_i(i_1)$};
    \fill[color4] (2.5,0) rectangle (2.9,0.35);
    \node[anchor=west, font=\small] at (3.0, 0.125) {$\mathbf{C}_j(j_2 \mid i_1)$};
    \fill[color6] (5.5,0) rectangle (5.9,0.35);
    \node[anchor=west, font=\small] at (6.0, 0.125) {$\mathbf{C}_k(k_1 \mid i_1, j_2)$};
 \end{tikzpicture}
\caption{Iteration sub-spaces enclosed by different cost functions over two $3\times 3 \times 2$ sparse tensors $A$ and $B$. Dotted coordinates are 0-valued and are not explicitly stored.}
\label{fig:coordinate-tree}
\end{figure}
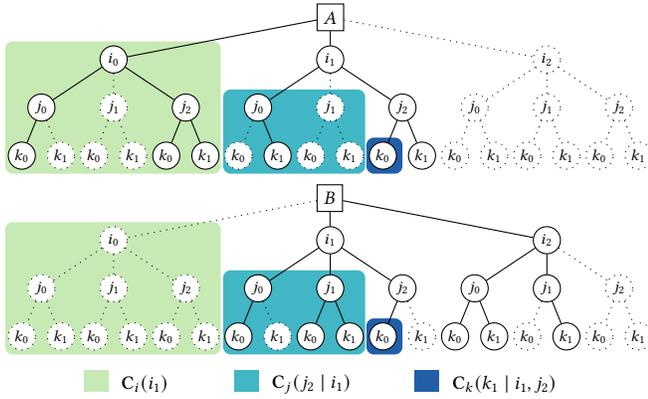
$\mathbf{C}_i(i_1)$ refers to all iteration space
points with the first coordinate in $[i_0, i_1)$, which is just coordinates starting with $i_0$.
Then, $\mathbf{C}_j(j_2 \mid i_1)$ refers to all iteration space points with $i_1$ as the $i$ coordinate and $j$ coordinates within $[0, j_2)$.
Note that points with $i$ coordinate as $i_0$ are excluded.
Likewise, $\mathbf{C}_k(k_1 \mid  i_1, j_2)$ is just the singular coordinate $(i_1, j_2, k_0)$, as $[k_0, k_1) = k_0$.

We require two properties of cost functions for use in our search algorithm:
\emph{monotonicity} and \emph{hierarchical consistency}: 
\begin{description}
    \item[Monotonicity:] Cost functions must be monotone:
\[
\WorkFunc{m}{x \mid x_0, \ldots, x_{m-1}} \leq \WorkFunc{m}{x+1 \mid x_0, \ldots, x_{m-1}}
\]
    \item[Hierarchical Consistency:] Cost functions must hierarchically compose, meaning
    the cost at a coordinate must equal the cost to fully coiterate the dimension after it.
    Specifically, $\forall~x_m \in [0, \mathbf{N}_m)$:
    \begin{align*}
    \WorkFunc{m}{x_m+1 \mid x_0, \ldots, x_{m-1}} =~ &\WorkFunc{m}{x_m \mid x_0, \ldots, x_{m-1}} + \\
                                                         &\WorkFunc{m+1}{N_{m+1} \mid x_0, \ldots, x_{m-1}, x_m}
    \end{align*}

\end{description}






\subsection{Hierarchical Search Partitioning}
\label{sec:hierarchical-search}

We now describe a partitioning algorithm that leverages these cost functions to find
load-balanced partitions of sparse tensor algebra operations over iteration spaces with
arbitrary dimensionality.
Our algorithm generalizes a one-dimensional binary search into a 
\textit{hierarchical} binary search with iterative (per-loop) partitioning.
The partitioning algorithm is shown in~\Cref{alg:hierarchical-search-v2}, which searches for the 
multi-dimensional coordinate that incurs a cost of $\mathcal{Q}$.
The hierarchical consistency property of cost functions enables the hierarchical search 
to be broken into a search at each iteration space dimension.
Likewise, monotonicity enables binary search within each dimension.
\Cref{alg:hierarchical-search-v2} is used to construct a partition for each parallel processor
by invoking it with a different value of $\mathcal{Q}$ for each processor.
Concretely, given total cost of $\mathcal{Q}^*$, processor $p \in [0, \mathcal{P})$ can discover partition boundaries
using~\Cref{alg:hierarchical-search-v2} by searching for $p \cdot \frac{\mathcal{Q}^*}{\mathcal{P}}$ as the partition
lower-bound and $(p +1) \cdot \frac{\mathcal{Q}^*}{\mathcal{P}}$ as the partition upper bound.


\begin{algorithm}
\caption{Hierarchical Search Partitioning Algorithm}
\small
\label{alg:hierarchical-search-v2}
\begin{algorithmic}[1]
\State \textbf{Input}:  Loop order $\mathcal{L}$, Cost functions $\mathbf{C}_{m}$ for $m \in \mathcal{L}$, Queried Cost $\mathcal{Q}$, Coordinate dimensions $N$
\State \textbf{Output}: Coordinate $\bar{x}$, tuple of per-dimension coordinates
\Function{FindPartition}{$\mathcal{L}$, $\mathbf{C}$, $\mathcal{Q}$, $N$}
    \State $\bar{x} = ()$,\, $\mathcal{R} = \mathcal{Q}$ \Comment{Init current coordinate and residual cost.}
    \For{$m$ in $\mathcal{L}$} \label{line:loops} \Comment{Loop over each iteration space dimension.}
        \State $\triangleright$ Binary search $x_m \in [0, N_m)$ with $\mathbf{C}_m(x_m \mid \bar{x},) \leq \mathcal{R}$.
        \State $x_{m}$ = \textsc{HighestCoordinateLeqQuery}($N_{m}$, $\mathcal{R}$, $\mathbf{C}_{m}$, $\bar{x}$) \label{line:leq-query}
        \State $\mathcal{R} \leftarrow \mathcal{R} - \mathbf{C}_{m}(x_{m} \mid \bar{x})$ \label{line:residual-work}
        \State $\bar{x} \leftarrow \bar{x} :: x_m$ \label{line:append-partition}
    \EndFor

    \State \textbf{return} $\bar{x}$
\EndFunction
\end{algorithmic}
\end{algorithm}

\Cref{alg:hierarchical-search-v2} restricts the search for partition boundaries from the full
$d$-dimensional iteration space to increasingly smaller sub-spaces by fixing coordinates
along each dimension.
This is done by iterating over the desired loop order and computing the highest coordinate that 
has a cost less than or equal to the running query cost, $\mathcal{R}$. 
The coordinate is included as a component of the result ($\bar{x}$), and its dimension's cost is subtracted from the running query.
This pattern maintains the invariant that $\mathcal{R}$ corresponds to the remaining query cost within the restricted subspace defined by the partially formed coordinate $\bar{x}$.
At the final level, our algorithm yields a coordinate tuple that precisely identifies a partition boundary
enclosing the queried cost $\mathcal{Q}$.


\subsubsection{Asymptotics}
\label{sec:alg-analysis}

Analysis of~\Cref{alg:hierarchical-search-v2} shows that it is
$O\Big(\sum_{l \in \mathcal{L}} \log(|l|) P_l\Big)$,
where $O(P_l)$ is the asymptotic complexity of evaluating the cost function at level $l$. 
With $n = \max_{l \in \mathcal{L}}(|l|)$ and $O(\mathcal{W})= \max_{l \in \mathcal{L}}(P_l)$ for 
some $\mathcal{W}$, this runtime simplifies to $O\left(|\mathcal{L}| \mathcal{W} \log(n) \right)$.
Our implementations of cost functions (\Cref{sec:format-work-functions}), are logarithmic in the number
of non-zero elements
and $|\mathcal{L}|$ is a constant for a given kernel, so the runtime is bounded by $O\left(\log(nnz) \log(n) \right)$.

\subsubsection{Tightness of Partitioning}
\label{sec:tightness}

\begin{theorem}
\label{thrm:tightness}
Let $\Delta_m$ denote the maximum cost difference between adjacent coordinates at loop level $m$:
\[
\Delta_m = \max\limits_{\forall\bar{x}_m, 0 \leq x_m < N_m}(\mathbf{C}_m(x_m+1 \mid \bar{x}_m) - \mathbf{C}_m(x_m \mid \bar{x}_m)) 
\]
Then for any cost query,~\Cref{alg:hierarchical-search-v2} returns $\bar{x} = (x_0,\ldots,x_{d-1})$
such that $0\leq\mathcal{Q} - \mathbf{C}(\bar{x}) < \Delta_{d-1}$ where $\mathbf{C}(\bar{x})$ is the cost of coiteration from coordinate $(0,\ldots,0)$ to $(x_0,\ldots,x_{d-1})$:
\[
\mathbf{C}(\bar{x}) = \sum_{m=0}^{d-1} \mathbf{C}_m(x_m \mid x_0,\ldots,x_{m-1}).
\]

\end{theorem}
\begin{proof}
By induction on $m$.
\end{proof}


\Cref{thrm:tightness} implies that for a partition defined by lower and upper bound cost query values
$\mathcal{Q}_l$ and $\mathcal{Q}_u$,~\Cref{alg:hierarchical-search-v2} returns 
coordinates $\bar{x}_l$ and $\bar{x}_u$ such that the cost of coiteration 
from $\bar{x}_l$ to $\bar{x}_u$ lies in the range 
$\mathcal{Q}_u - \mathcal{Q}_l \pm \Delta_{d-1}$.
Consequently, partitions constructed from 
equally spaced query values incur a maximum load imbalance of at most 
$2\cdot\Delta_{d-1}$. 
Furthermore, our implementation of $\mathbf{C}$ 
(\Cref{sec:format-work-functions}) ensures that $\Delta_{d-1}$ is small (bounded by the number of tensors being coiterated), 
thereby guaranteeing that parallel processors execute approximately equal amounts of work.


\subsection{Recursive Partitioning}
\label{sec:recursive-partitioning}

\Cref{sec:hierarchical-search} describes how to compute partitions that balance the number of 
non-zeros processed in each partition equally across parallel workers.
%
However, the hierarchical algorithm is insufficient for partitioning nested intersections on sparse tensors.
This is because the hierarchical search algorithm and cost functions assume that all 
non-zeros within a candidate coordinate range are traversed, which is true only when unions are being 
computed or if a sparse intersection is performed only in the innermost loop.
When iteration over nested sparse iterations is required, such as in the Hadamard product
of DCSR matrices (\Cref{lst:emul-dcsr2-cfir}), entire segments of the
iteration space are skipped when higher-order coordinates in an intersected dimension do not match, as illustrated in~\Cref{fig:dcsr-mul-partition}.
Nested sparse intersections require changing how costs are computed to exclude the cost
of hierarchically skipped portions of the iteration space.



Because the exact coordinates that are skipped during iteration with intersections are only
known dynamically, intersections in outer loops must be computed first to discover the corresponding
costs of intersections within inner loops.
Therefore, we propose an iterative approach that progressively computes new cost functions on the
remaining portions of the coordinate space after discovering which coordinates can be skipped.
Furthermore, these functions can be computed while coiterating through each sparse intersection
in a load-balanced manner.


\newcommand{\registernode}[1]{%
  \expandafter\def\csname present@#1\endcsname{1}%
}

\newcommand{\nodewithlabel}[7]{%
  \if1#1%
    \node[
      draw, circle,
      minimum size=4.5mm,
      inner sep=0.5pt,
      font=\scriptsize,
      fill=#7
    ] (#6) at (#2,#3) {$#4_{#5}$};%
  \else
    \node[
      draw, circle, dotted,
      minimum size=4.5mm,
      inner sep=0.5pt,
      font=\scriptsize,
      fill=white,
    ] (#6) at (#2,#3) {$#4_{#5}$};%
  \fi
}

\newcommand{\drawcoordinatetree}[9]{%
    \begingroup
  \forcsvlist{\registernode}{#9}

  \if4#1
    \node[anchor=west, font=\small, draw, rectangle] (#4) at (1.6, 0.8) {$#4$};
  \else
    \node[anchor=west, font=\small, draw, rectangle] (#4) at (-1.1 + #7, 0.8) {$#4$};
  \fi

  \if4#1
        \fill[color2]
            (-2.7,-1.1) -- (-2.7,-0.4) -- (-0.9,-0.4) -- (-0.9,-1.1) -- cycle;
        \fill[color4]
            (2.7,-1.1) -- (2.7,-0.4) -- (0.9,-0.4) -- (0.9,-1.1) -- cycle;
    \else 
        \fill[color2]
            (-2.7+#7,-1.1) -- (-2.7+#7,0.3) -- (-0.9+#7,0.3) -- (-0.9+#7,-1.1) -- cycle;
        \fill[color4]
            (0.9+#7,-1.25) -- (0.9+#7,0.45) -- (-0.9+#7,0.45) -- (-0.9+#7,-1.25) -- cycle;
        \fill[color4]
            (-2.85+#7,0.3) -- (-2.85+#7,0.45) -- (-0.9+#7,0.45) -- (-0.9+#7,0.3) -- cycle;
        \fill[color4]
            (-2.85+#7,-1.1) -- (-2.85+#7,-1.25) -- (-0.9+#7,-1.25) -- (-0.9+#7,-1.1) -- cycle;
        \fill[color4]
            (-2.85+#7,0.3) -- (-2.85+#7,-1.1) -- (-2.7+#7,-1.1) -- (-2.7+#7,0.3) -- cycle;
    \fi

  \foreach \ii in {0,...,#1} {
    \pgfmathsetmacro{\xI}{(\ii-1)*#5 + #7}

    \xdef\Ipresent{0}
    \foreach \jj in {0,...,#2} {
      \foreach \kk in {0,...,#3} {
        \ifcsname present@\ii/\jj/\kk\endcsname
          \xdef\Ipresent{1}
        \fi
      }
    }
    \edef\ilabel{i'}
    \if4#1
        \xdef\ilabel{i}
    \fi

    \nodewithlabel{\Ipresent}{\xI}{0}{\ilabel}{\ii}{inode\ii}
    {white}

    \if1\Ipresent
        \draw[thin, black] (#4) -- (inode\ii);
    \else
        \draw[thin,black, dotted] (#4) -- (inode\ii);
    \fi

    \foreach \jj in {0,...,#2} {
      \pgfmathsetmacro{\xJ}{\xI + (\jj - (#2)/2.0)*#6}
      \pgfmathsetmacro{\yJ}{-#8}

      \xdef\Jpresent{0}
      \foreach \kk in {0,...,#3} {
        \ifcsname present@\ii/\jj/\kk\endcsname
          \xdef\Jpresent{1}
        \fi
      }

      \if1\Ipresent
          \if1\Jpresent
            \draw[thin, black] (inode\ii) -- (\xJ, \yJ);
          \else
            \draw[thin,black, dotted] (inode\ii) -- (\xJ, \yJ);
          \fi
      \else
        \draw[thin,black, dotted] (inode\ii) -- (\xJ, \yJ);
        \fi
    
    \nodewithlabel{\Jpresent}{\xJ}{\yJ}{j}{\jj}{jnode\ii-\jj}{white}
    }
  }

    \endgroup
    



}

\newcommand{\drawtwotreesB}[9]{%
  \drawcoordinatetree
    {#1}
    {#2}
    {#3}
    {#4}
    {#5}
    {#6}
    {#7}
    {#8}
    {#9}
}
\newcommand{\drawtwotrees}[9]{%
  \drawcoordinatetree
    {#1}
    {#2}
    {#3}
    {#4}
    {#5}
    {#6}
    {#7}
    {#8}
    {#9}
  \vspace{12pt}%
  \drawtwotreesB%
}

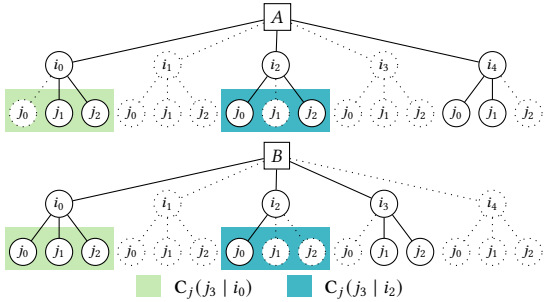
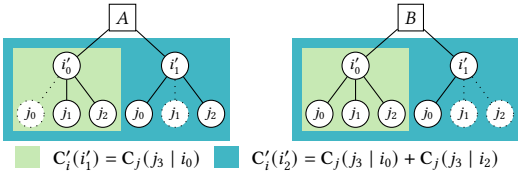
\begin{figure}[b]
\begin{subfigure}{0.46\textwidth}
\centering
\begin{tikzpicture}[scale=0.8, transform shape]
  \drawcoordinatetree
  {4}{2}{1}{A}{1.8}{0.6}{0}{0.8}
  {0/1/0, 0/2/0, 2/0/0, 2/2/0, 4/0/0, 4/1/0}
\end{tikzpicture}

\vspace{15pt}

\begin{tikzpicture}[scale=0.8, transform shape]
  \drawcoordinatetree
  {4}{2}{1}{B}{1.8}{0.6}{0}{0.8}
  {0/0/0, 0/1/0, 0/2/0, 2/0/0, 3/1/0, 3/2/0}
\end{tikzpicture}

\begin{tikzpicture}[scale=0.8, transform shape]
        \fill[color2] (0,0) rectangle (0.4,0.35);
        \node[anchor=west, font=\small] at (0.5, 0.125) {$\mathbf{C}_j(j_3 \mid i_0)$};
        \fill[color4] (2.5,0) rectangle (2.9,0.35);
        \node[anchor=west, font=\small] at (3.0, 0.125) {$\mathbf{C}_j(j_3 \mid i_2)$};
\end{tikzpicture}

\subcaption{The loop nest $i,j$ contains a sparse intersection at outer level $i$. We apply recursive partitioning to map $i, \mathbf{C}_i$ to $i', \mathbf{C'}_{i}$.}
\end{subfigure}%
\vfill
\vspace{1em}
\begin{subfigure}{0.46\textwidth}
\centering
\begin{tikzpicture}[scale=0.8, transform shape]
  \drawcoordinatetree
  {1}{2}{1}{A}{1.8}{0.6}{-2.4}{0.8}
  {0/1/0, 0/2/0, 1/0/0, 1/2/0}
  \drawcoordinatetree
  {1}{2}{1}{B}{1.8}{0.6}{2.4}{0.8}
  {0/0/0, 0/1/0, 0/2/0, 1/0/0}
\end{tikzpicture}
\begin{tikzpicture}[scale=0.8, transform shape]
        \fill[color2] (0,0) rectangle (0.4,0.35);
        \node[anchor=west, font=\small] at (0.5, 0.125) {$\mathbf{C}_i'(i'_1) = \mathbf{C}_j(j_3 \mid i_0)$};
\end{tikzpicture}
\begin{tikzpicture}[scale=0.8, transform shape]
        \fill[color4] (2.5,0) rectangle (2.9,0.35);
        \node[anchor=west, font=\small] at (3.0, 0.125) {$\mathbf{C}_i'(i'_2) = \mathbf{C}_j(j_3 \mid i_0) + \mathbf{C}_j(j_3 \mid i_2)$};
\end{tikzpicture}
\subcaption{Modified loop nest $i', j$ does not have hierarchical skipping. The remapped cost function $\mathbf{C'}_{i}$ is calculated from $\mathbf{C}_j$.}
\end{subfigure}
\caption{Recursive partitioning for Hadamard product. 
}
\label{fig:coordinate-tree-recurssion}
\end{figure}

\subsubsection{Remapped Cost Functions}
\sys{} reasons about hierarchical skipping by \emph{remapping} cost functions
into cost functions that are aware of intersection's effects on coiteration.
A remapped cost function considers the cost of coiteration over a given subset of the coordinate space.
For example, a remapped cost function for~\Cref{lst:emul-dcsr2-cfir} can be computed over the valid 
$i$'s that are in $\texttt{rows}(A) \, \cap \, \texttt{rows}(B)$, where $\mathds{1}_p$ denotes the indicator function
over the predicate $p$:
\[
    \mathbf{C}'_{i}(x_{i}) = \sum_{r=0}^{x_i} \mathbf{C}_{j}(N_{j}\mid r) \cdot \mathds{1}_{r \in \texttt{rows}(A) \, \cap \, \texttt{rows}(B)}
\]
$\mathbf{C}'_{i}$ can be used to load-balance~\Cref{lst:emul-dcsr2-cfir} despite the hierarchical skipping: it only considers the cost of coordinate sub-spaces that must be coiterated over after the outermost intersection, precisely
describing the costs of iteration in the inner loop, as illustrated in~\Cref{fig:coordinate-tree-recurssion}.
%
More generally, we define a remapped cost function that considers iteration over a subset of the dimension $m$ given by the set $I$:
\[
    \mathbf{C}_{m}'\left(x_m \mid \bar{x}_m\right) = \sum_{r=0}^{x_m} \mathbf{C}_{m+1}(N_{m+1} \mid \bar{x}_m, r)\cdot\mathds{1}_{r \in I}
\]

\subsubsection{Efficiently Computing Remapped Cost Functions}
The sum in a remapped cost function over the subset $I$ can be memoized by iterating $I$ and independently computing the cost function $\mathbf{C}_{m+1}$ for each valid coordinate, then taking the prefix sum of the stored values.
For example,~\Cref{lst:emul-dcsr2-cfir} is rewritten into~\Cref{lst:emul-dcsr2-rewritten}, where each loop is load-balanced.
The prefix sum can be computed efficiently in parallel~\cite{BlellochTR90, merrill2016single}.

\begin{figure}
    \centering
\centering
\begin{lstlisting}[language=cfir, mathescape=true, basicstyle=\footnotesize\ttfamily, backgroundcolor=\color{codebackground}, frame=single,
caption={Computing $Z_{ij} = A_{ij}\cdot B_{ij}$ on DCSR matrices via multiple load-balanced kernels.},
label={lst:emul-dcsr2-rewritten},
numbers=left,
xleftmargin=1.5em,
escapeinside={(*}{*)} % Allows LaTeX commands inside the code
]
// Load balance on just loop $i$.
while $i~\shortleftarrow$ rows($A$) $\cap$ rows($B$): (*\label{line:loop-i-while}*)
  $T[i]$ = $\mathbf{C}_{j}(N_{j} \mid i)$ (*\label{line:assignT}*)
$T'$ = exclusive_prefix_sum($T$)  (*\label{line:sumT}*)
let $\mathbf{C}_{i}'(x_i)$ = $T'[x_i]$  (*\label{line:assignW}*)
// Load balance using $\mathbf{C}_{i}'$ and $\mathbf{C}_{j}$.
for $i~\shortleftarrow~i_T$:  (*\label{line:loop-i-for}*)
  while $j~\shortleftarrow$ cols($A[i]$) $\cap$ cols($B[i]$): (*\label{line:loop-j}*)
    $Z[i, j]$ = $A[i, j]$ * $B[i, j]$ (*\label{line:assignZ}*)
\end{lstlisting}
\end{figure}

\subsubsection{Applying Recursive Partitioning}

The structure of recursive partitioning is given in~\Cref{alg:recursive-search-alg}, where the base case (\Cref{alg:hierarchical-search-v2}) is applied if a loop nest contains no sparse intersections or only an innermost sparse intersection. 
If an outer sparse intersection is found, the body of the sparse intersection is replaced with the computation of the remapped cost function.
This modified loop nest is partitioned and computed in parallel, giving a precise cost model for the partitioning of the body of the loop. 
This recursive structure results in a series of load-balanced parallel loops.

\begin{algorithm}
\caption{Recursive Partitioning Algorithm}
\small
\label{alg:recursive-search-alg}
\begin{algorithmic}[1]
\State \textbf{Input}:  Loop order $\mathcal{L}$, Cost Functions $\mathbf{C}_{m}$ for $m \in \mathcal{L}$, Cost $\mathcal{Q}$
\State \textbf{Output}: Coordinate $\bar{x}$, tuple of per-dimension coordinates
\Function{FindRecursivePartition}{$\mathcal{L}$, $\mathbf{C}$, $\mathcal{Q}$}
    \State $m = \textsc{FindOutermostSparseIntersection}(\mathcal{L})$
        
    \If{$m$ is not \texttt{null} and $m$ is not innermost loop}
        \State \textcolor{blue}{$\triangleright$ Hierarchical skipping, remap cost function.}
        \State $\mathbf{C}'_m, m'$ = $\textsc{MapIntersectionToNewSpace}$($C_m$,
        $m$)
        \State $\mathbf{C}[m] = \mathbf{C}'_m$,  $\mathcal{L}[m] = m'$
        \State  $\textbf{return}$ \textsc{FindRecursivePartition}($\mathcal{L}$, $\mathbf{C}$, $\mathcal{Q}$)
    \Else
        \State \textcolor{blue}{$\triangleright$ No hierarchical skipping, do standard partitioning.}
        \State $\textbf{return}$ $\textsc{FindPartition}$($\mathcal{L}$, $\mathbf{C}$, $\mathcal{Q}$)
    \EndIf
\EndFunction
\end{algorithmic}
\end{algorithm}





\section{Code Generation}
\label{sec:codegen}

\sys{} is implemented as an extension of the \taco compiler that uses the partitioning algorithms
in~\Cref{sec:partition} to efficiently parallelize coiterative loops.
Our points of extension are illustrated in~\Cref{fig:nacho-lowering}. 
\sys generates three parallel kernels per load-balanced loop nest: a \emph{partitioning} kernel, an 
\textit{assembly} kernel that allocates the data structures for sparse result tensors, and a 
\textit{compute} kernel that fills in the sparse output allocated by the assembly kernel.
Both the assembly and compute kernels use the partitions computed by the partitioning kernel.
We primarily build off of \taco~\cite{kjolstad2017taco, chou2018formats, chou2020conversion}
and provide the relevant background to understand our contributions. 
We refer the reader to these works or to \citet{kjolstad2020thesis} for more thorough descriptions of \taco.

We first briefly describe the requirements of a parallel backend, and provide our technique for generating parallel partitioning kernels in~\Cref{sec:par-partition} with our implementation of data-structure-specific cost functions in~\Cref{sec:format-work-functions}. 
Lastly, we show how the computed partitions are used for both the assembly and compute phases
in~\Cref{sec:par-assembly,sec:par-compute}.

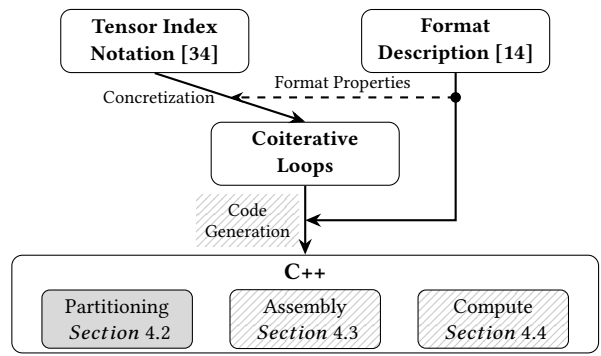
\begin{figure}[b]
    \centering
\begin{tikzpicture}[
  mainbox/.style = {
    draw, rectangle, rounded corners = 4pt,
    minimum width  = 2.5cm,
    minimum height = 0.75cm,
    align  = center,
    font   = \footnotesize\bfseries
  },
  subbox/.style = {
    draw, rectangle, rounded corners = 3pt,
    minimum width  = 2.0cm,
    minimum height = 0.7cm,
    align  = center,
    font   = \footnotesize
  },
  arr/.style  = {-{Stealth[length=6pt, width=5pt]}, thick},
  darr/.style = {-{Stealth[length=5pt, width=4pt]}, thick, dashed},
  lbl/.style  = {font = \scriptsize, align = center}
]

  \node[mainbox] (TIN) at (-2.0, 0)
    {Tensor Index\\Notation~\cite{kjolstad2017taco}};
  \node[mainbox] (FD)  at ( 2.0, 0)
    {Format\\Description~\cite{chou2018formats}};

  \node[mainbox] (CFIR) at (0, -1.5) {Coiterative\\Loops};

  \node[draw, rectangle, rounded corners = 4pt,
        minimum width = 7.8cm, minimum height = 1.3cm]
        (CPP) at (0, -3.5) {};
  \node[font = \small\bfseries, anchor = north, inner sep = 3pt]
        at (CPP.north) {C++};
  \node[subbox, fill=gray!30]
        (PAR)  at ($(CPP.center) + (-2.5, -0.2)$) {Partitioning\\$~\Cref{sec:par-partition}$};
  \node[subbox, pattern=north east lines, pattern color=gray!40]
        (ASM)  at ($(CPP.center) + ( 0.0, -0.2)$) {Assembly\\$~\Cref{sec:par-assembly}$};
  \node[subbox, pattern=north east lines, pattern color=gray!40]
        (COMP) at ($(CPP.center) + ( 2.5, -0.2)$) {Compute\\$~\Cref{sec:par-compute}$};

  
  \draw[arr] (TIN.south) -- coordinate[midway] (conc-mid)
    node[lbl, left, midway, xshift=-2pt] {Concretization}
    (CFIR.north);

  \draw[arr] (CFIR.south) -- coordinate[midway] (cg-mid)
    node[lbl, left, midway, xshift=-2pt, pattern=north east lines, pattern color=gray!40] {Code\\Generation}
    (CPP.north);

  \draw[arr] (FD.south) |- (cg-mid);

  \coordinate (fd-branch) at (FD.south |- conc-mid);
  \fill (fd-branch) circle (2pt);
  \draw[darr] (fd-branch) -- 
    node[lbl, above, midway, yshift=-2pt] {Format Properties} 
    (conc-mid);

\end{tikzpicture}
    \caption{
    \sys code generation. Gray boxes denote modifications to prior work in \taco{}. We intercept the lowering of coiterative loops to generate partition functions and modify the iteration bounds used in compute and assembly stages.
    }
    \label{fig:nacho-lowering}
\end{figure}

\subsection{Backend-Specific Code Generation}
\label{sec:codegen-backend}

In the following sections we describe a parallel backend assuming a shared-memory
architecture and five capabilities: parallel loops, bulk memory allocation, prefix sums, and for scatter kernels (see~\Cref{sec:par-assembly}), sorting and segmented summation.
These are standard capabilities provided by vendor libraries~\cite{cccl, onetbb}.
For segmented summation, we use a custom load-balanced implementation based on moderngpu~\cite{baxter2016moderngpu}.
\IGNORE{
\begin{description}
    \item[Parallel Loops:] Implemented by CUDA kernels and \texttt{tbb}'s \texttt{parallel\_for}.

    \item[Bulk Memory Allocation:] \texttt{cudaMallocAsync} for GPU and \texttt{malloc} for CPU.

    \item[Prefix Sum:] Implemented by \texttt{cub}'s \texttt{ExclusiveSum}~\cite{merrill2016single, cccl} for GPUs and and \texttt{std::exclusive\_scan} for CPUs.

    \item[Sorting:] We use cub's \texttt{DeviceSegmentedSort}~\cite{cccl} for GPUs and \texttt{tbb}'s \texttt{parallel\_sort} on CPUs.

    \item[SegmentedSum:] \TODO{}
\end{description}
}
We believe \sys can be extended to other parallel backends that support these capabilities.


\subsection{Parallel Partitioning}
\label{sec:par-partition}

To produce a partitioning kernel, \sys specializes~\Cref{alg:hierarchical-search-v2} with a concrete loop order
(selected with the algorithm of~\citet{yan2026scorch}) and cost functions generated with respect to the
data structure formats of the input tensors.
%
%
Specifically, we implement cost functions as the sum of the number of non-zero values contained by each tensor operand up to the query coordinate.
This choice of cost function satisfies both the monotonicity and consistency properties described in ~\Cref{sec:work-funcs} as requirements for a cost function, and is a good proxy for work, as it closely corresponds to the number of memory reads per thread.
%
%
The number of non-zeros up to a coordinate for a specific operand is computed via data-structure-specific queries, which we describe in~\Cref{sec:format-work-functions}.
As part of specialization, we perform an optimization to search over sparse tensors
and store their partitions by using the \emph{position} space (iterators into sparse
data structures), even though the searches are logically over the \emph{coordinate} space.
%
We first provide a brief background on \taco formats, as they relate to our code generation.

\subsubsection{Background: \textsc{Taco} Formats}

\taco decouples compute specifications from a \textit{format} language that describes the
data structure used to store each tensor dimension~\cite{chou2018formats}.
For example, the Compressed Sparse Row (CSR)~\cite{tinney1967csr} format is expressed as a Dense-Compressed row-major matrix.
%
~\citet[Section 2]{chou2018formats} surveys tensor storage formats.

\subsubsection{\sys Formats}

\sys supports the Compressed, Dense, and Singleton formats. 
These are \textit{sorted} formats, which enable \sys to perform binary searches over the data structures.
We believe that our choice of cost function is applicable to a wider variety of formats, such as those described by~\citet{ahrens2025finch}, as long as those formats support the efficient calculation of the number
of non-zero values within a range of coordinates.
Note that unsorted formats (such as hash-based or unsorted COO) that do not enable coiteration as described in ~\Cref{sec:coiteration-background} can still be partitioned via Nacho's technique, but the partitioning step is not efficient due to the required linear scans per cost model query instead of binary searches.
Additionally, iterating data structures in \textit{discordant} order, for example, iterating a CSR by columns before rows (CSR is a row-major format), is not supported by \taco, and therefore also not supported by \sys.

\subsubsection{Format-Specific Cost Functions}
\label{sec:format-work-functions}

Our cost function implementations use the sum of the number of non-zeros contained
in each sparse tensor within the queried coordinate subspace, which satisfies the
properties defined in~\Cref{sec:work-funcs}.
%
%
This approximates the runtime of each partition, but future work could consider cardinality estimators~\cite{haas1996cardinality, salihoglu2026cardinalitygraphs, vengerov2015joinsizeestimation} or 
a finer-grained estimation of memory traffic. 


\begin{figure}[t]
    \centering
    \begin{minipage}{0.48\textwidth}
\begin{lstlisting}[language=cpp, mathescape=true, backgroundcolor=\color{codebackground}, frame=single,
caption={Generated cost functions for CSR matrix $A_{ij}$ over loop nest $i\to j$.},
label = {lst:csr-work-functions}% , numbers=left
]
int C_i(int x_i, CSR A) {
 return A.pos[x_i] - A.pos[0];
}
int C_j(int x_i, int x_jp, CSR A) {
 return x_jp - A.pos[x_i];
}
\end{lstlisting}
\end{minipage}
\vfill
\begin{minipage}{0.48\textwidth}
\begin{lstlisting}[language=cpp, mathescape=true, backgroundcolor=\color{codebackground}, frame=single,
caption={Generated cost functions for CSR matrix $A_{ij}$ over loop nest $i\to k \to j$. $A$ is a broadcasted over level $k$.}, 
label={lst:csr-work-functions-broadcast}% , numbers=left
]
int C_i(int x_i, CSR A) {
 return (A.pos[x_i] - A.pos[0]) * N_k;
}
int C_k(int x_i, int x_k, CSR A) {
 return (A.pos[x_i+1] - A.pos[x_i]) * x_k
}
int C_j(int x_i, int x_k, int x_jp, CSR A) {
 return x_jp - A.pos[x_i]; // Same as $\text{\Cref{lst:csr-work-functions}}$
}
\end{lstlisting}
\end{minipage}
\end{figure}

\sys{} generates a cost function for each level of each sparse tensor in the compiled tensor expression.
Our code generation uses the query scheme described by~\citet{chou2020conversion} to generate
format-specific implementations to count the number of non-zeros included within the cost function query.
%
%
Instead of accepting coordinates, the cost functions instead accept \emph{iterators} into the
dimensions of sparse tensors that reference the physical locations within the sparse data
structures referring to a logical coordinate.
%
\Cref{lst:csr-work-functions} shows examples for the CSR data structure:
the cost function for the outer dense level returns the number of non-zeros in the level below it,
while the cost function for the compressed inner level subtracts the offset of the start of that row from the iterator that corresponds to the coordinate.

When an iteration space dimension does not appear in a tensor's index expression,
such as the reference to $x_j$ when iterating over the $i$ dimension of
the expression $y_i = A_{ij} \cdot x_{j}$, we refer to the tensor as \emph{broadcasted}
over that dimension.
Broadcasted dimensions induce repeated coiteration over the tensor, which
affects cost estimation.
As illustrated in~\Cref{lst:csr-work-functions-broadcast}, \sys{} generates cost functions for
broadcasted dimensions that scale the cost by the appropriate factor.

\subsubsection{Specializing the Search for Partitions}
\sys{} generates a specialized implementation of~\Cref{alg:hierarchical-search-v2}
for the target expression and sparse tensor data structures.
\sys{} inlines cost function invocations and unrolls the main loop over 
the ordering of iteration space dimensions.
The result of specializing~\Cref{alg:hierarchical-search-v2} for the three-way sparse 
vector co-iteration example (\Cref{fig:vec_lb_example}) is shown in \Cref{lst:partition-pos-space}.
The generated code contains an outer binary search (Line~\ref{listing:outer-binary-search}) for the $i$ coordinate 
with cost less than the input query $q$.
Each step of the outer binary search then invokes the cost function for each sparse vector,
which results in inner binary searches over the non-zero elements in each sparse vector (lines~\ref{listing:cost-func-begin}---\ref{listing:cost-func-end}).
\sys{} also tracks the searched position ranges within each
compressed data structure to limit the range of positions searched in each inner 
binary search as the outer binary search progresses (lines~\ref{listing:pos-track-begin}---\ref{listing:pos-track-end}).
Finally, \sys{} stores the results of the partitioning search so that
the results can be reused in the assembly and compute stages (line~\ref{listing:save-pos-ranges}).

\begin{figure}[t]
\centering
\begin{lstlisting}[language=cpp, mathescape=true, frame=single, backgroundcolor=\color{codebackground}, 
numbers=left,
xleftmargin=1em,
caption={\Cref{alg:hierarchical-search-v2} specialized to coiteration of three sparse vectors, which generates the partitions in \Cref{fig:vec_lb_example}.},
label={lst:partition-pos-space},
xleftmargin=2em
]
int partition(SpVec A, SpVec B, SpVec C, int q,
              int N, Parts &p) {
 int low = 0, high = N, lA = 0, lB = 0, lC = 0;
 int hA = A.nnz-1, hB = B.nnz-1, hC = C.nnz-1;
 while (low < high) {$\label{listing:outer-binary-search}$
  int i = low + (high - low + 1) / 2;
  // Find least ipA in [lA, hA] s.t. A[ipA] >= i. $\label{listing:cost-func-begin}$
  int ipA = lb_search(A, i, lA, hA);
  int ipB = lb_search(B, i, lB, hB);
  int ipC = lb_search(C, i, lC, hC);
  // Compute cost = C_i(i,A) + C_i(i,B) + C_i(i,C).
  int cost = ipA + ipB + ipC;$\label{listing:cost-func-end}$
  if (cost <= q)$\label{listing:pos-track-begin}$
   low = i, lA = ipA, lB = ipB, lC = ipC;
  else 
   high = i - 1, hA = ipA, hB = ipB, hC = ipC;$\label{listing:pos-track-end}$
 }
 p.ipA = lA, p.ipB = lB, p.ipC = lC;$\label{listing:save-pos-ranges}$
 return low;
}
\end{lstlisting}
\end{figure}

\subsubsection{Partitioning Kernel per Sparse Intersection}
A partitioning kernel is generated once per loop nest isolated by~\Cref{alg:recursive-search-alg}.
For example, a CSR element-wise operation produces the kernels illustrated in~\Cref{fig:union_kernels}, but a DCSR element-wise multiplication (nested sparse intersections, e.g., \Cref{lst:emul-dcsr2-cfir}), produces two partitioning kernels, one for load balancing each sparse intersection.
We illustrate the flow of kernels for DCSR element-wise multiplication in~\Cref{fig:dcsr_kernels}.
This load balances each stage of~\Cref{lst:emul-dcsr2-rewritten}: the first partitioning kernel load balances the outer intersection (Lines~\ref{line:loop-i-while}--\ref{line:assignT}), and produces partitions for the assembly and compute kernels to iterate over the outer intersection and compute the remapped cost function.
The second partitioning kernel uses the remapped cost function (Lines~\ref{line:sumT}--\ref{line:assignW}) to load balance the inner intersection (Lines~\ref{line:loop-i-for}--\ref{line:assignZ}), producing partitions for the final assembly and compute functions.

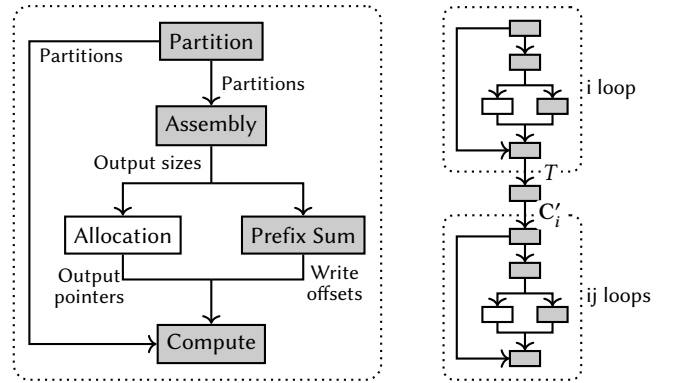
\begin{figure}[tpb]
    \centering

    \begin{subfigure}[b]{0.58\linewidth} 
        \centering
        \resizebox{\linewidth}{!}{
        \begin{tikzpicture}[
            node distance=0.8cm,
            box/.style={rectangle, draw, thick, minimum width=1cm, minimum height=0.5cm, font=\small\sffamily},
            arrow/.style={->, thick}
        ]
            \node[box, fill=gray!40] (partition) {Partition};
            \node[box, below=0.6cm of partition, fill=gray!40] (assembly) {Assembly};

            \coordinate[below=1.2cm of assembly.south] (mid);
            \node[box, left=0.4cm of mid, anchor=east] (alloc) {Allocation};
            \node[box, right=0.4cm of mid, anchor=west, fill=gray!40] (prefix) {Prefix Sum};

            \node[box, below=1.2cm of mid, fill=gray!40, anchor=north] (compute) {Compute};

            \draw[arrow] (partition) -- node[right, font=\sffamily\footnotesize] {Partitions} (assembly);

            \coordinate (split) at ([yshift=-0.5cm]assembly.south);
            \draw[thick] (assembly.south) -- node[left, font=\sffamily\footnotesize] {Output sizes} (split);
            \draw[arrow] (split) -| (alloc.north);
            \draw[arrow] (split) -| (prefix.north);

            \coordinate (merge) at ([yshift=0.6cm]compute.north);
            \draw[thick] (alloc.south) -- node[left, font=\sffamily\footnotesize, align=center, yshift=-8pt, xshift=2pt, inner sep=1pt] {Output\\[-2pt]pointers} (alloc.south |- merge) -- (merge);
            \draw[thick] (prefix.south) -- node[right, font=\sffamily\footnotesize, align=center, yshift=-6pt, inner sep=1pt] {Write\\[-2pt]offsets} (prefix.south |- merge) -- (merge);
            \draw[arrow] (merge) -- (compute.north);

            \draw[arrow] (partition.west) -- ++(-1.75,0) coordinate (p1)
                         -- node[right, font=\sffamily\footnotesize, yshift=53pt] {Partitions} (p1 |- compute.west)
                         -- (compute.west);

            \begin{scope}[on background layer]
                \node[draw, dotted, thick, rounded corners=6pt,
                      inner sep=6pt,
                      fit={(partition) (assembly) (alloc) (prefix) (compute) (p1)}] {};
            \end{scope}
        \end{tikzpicture}%
        }
        \caption{Sparse union or innermost-loop sparse intersection.}
        \label{fig:union_kernels}
    \end{subfigure}\hfill
    \begin{subfigure}[b]{0.4\linewidth} 
        \centering
        \hspace*{-1em}
        \begin{tikzpicture}[scale=0.89, 
            node distance=0.8cm,
            box/.style={rectangle, draw, thick, minimum width=1cm, minimum height=0.5cm, font=\small\sffamily},
            arrow/.style={->, thick}
        ]
            \begin{scope}[scale=0.45, transform shape]
                \node[box, fill=gray!40, anchor=north] (mpart_i) at (0,0) {};
                \node[box, below=0.6cm of mpart_i, fill=gray!40] (masst_i) {};

                \coordinate[below=1.2cm of masst_i.south] (mmid_i);
                \node[box, left=0.4cm of mmid_i, anchor=east] (malloc_i) {};
                \node[box, right=0.4cm of mmid_i, anchor=west, fill=gray!40] (mpref_i) {};

                \node[box, below=1.2cm of mmid_i, fill=gray!40, anchor=north] (mcomp_T) {};

                \draw[arrow] (mpart_i) -- (masst_i);
                \coordinate (msplit_i) at ([yshift=-0.5cm]masst_i.south);
                \draw[thick] (masst_i.south) -- (msplit_i);
                \draw[arrow] (msplit_i) -| (malloc_i.north);
                \draw[arrow] (msplit_i) -| (mpref_i.north);
                \coordinate (mmerge_i) at ([yshift=0.6cm]mcomp_T.north);
                \draw[thick] (malloc_i.south) -- (malloc_i.south |- mmerge_i) -- (mmerge_i);
                \draw[thick] (mpref_i.south) -- (mpref_i.south |- mmerge_i) -- (mmerge_i);
                \draw[arrow] (mmerge_i) -- (mcomp_T.north);
                
                \draw[arrow] (mpart_i.west) -- ++(-1.75,0) coordinate (mp1)
                             -- (mp1 |- mcomp_T.west)
                             -- (mcomp_T.west);

                \node[box, below=0.9cm of mcomp_T, fill=gray!40] (mpref_T) {};
                \draw[arrow] (mcomp_T.south) -- (mpref_T.north);

                \node[box, below=0.9cm of mpref_T, fill=gray!40] (mpart_ij) {};
                \node[box, below=0.6cm of mpart_ij, fill=gray!40] (masst_ij) {};

                \coordinate[below=1.2cm of masst_ij.south] (mmid_ij);
                \node[box, left=0.4cm of mmid_ij, anchor=east] (malloc_j) {};
                \node[box, right=0.4cm of mmid_ij, anchor=west, fill=gray!40] (mpref_j) {};

                \node[box, below=1.2cm of mmid_ij, fill=gray!40, anchor=north] (mcomp_ij) {};

                \draw[arrow] (mpart_ij) -- (masst_ij);
                \coordinate (msplit_ij) at ([yshift=-0.5cm]masst_ij.south);
                \draw[thick] (masst_ij.south) -- (msplit_ij);
                \draw[arrow] (msplit_ij) -| (malloc_j.north);
                \draw[arrow] (msplit_ij) -| (mpref_j.north);
                \coordinate (mmerge_ij) at ([yshift=0.6cm]mcomp_ij.north);
                \draw[thick] (malloc_j.south) -- (malloc_j.south |- mmerge_ij) -- (mmerge_ij);
                \draw[thick] (mpref_j.south) -- (mpref_j.south |- mmerge_ij) -- (mmerge_ij);
                \draw[arrow] (mmerge_ij) -- (mcomp_ij.north);
                
                \draw[arrow] (mpart_ij.west) -- ++(-1.75,0) coordinate (mp2)
                             -- (mp2 |- mcomp_ij.west)
                             -- (mcomp_ij.west);

                \draw[arrow] (mpref_T.south) -- (mpart_ij.north);

                \begin{scope}[on background layer]
                    \node[draw, dotted, thick, rounded corners=4pt,
                          inner sep=5pt,
                          fit={(mpart_i) (masst_i) (malloc_i) (mpref_i) (mcomp_T) (mp1)}] (box_i) {};
                    \node[font=\Huge\sffamily, anchor=west, overlay] at (box_i.east) {i loop};
                    
                    \node[draw, dotted, thick, rounded corners=4pt,
                          inner sep=5pt,
                          fit={(mpart_ij) (masst_ij) (malloc_j) (mpref_j) (mcomp_ij) (mp2)}] (box_ij) {};
                    \node[font=\Huge\sffamily, anchor=west, overlay] at (box_ij.east) {ij loops};
                \end{scope}
            \end{scope}

            \node[right, font=\sffamily\small, inner sep=1pt, xshift=6pt,
                  fill=white, rounded corners=1pt]
                at ($(mcomp_T.south)!0.5!(mpref_T.north)$) {$T$};
            \node[right, font=\sffamily\small, inner sep=1pt, xshift=4pt,
                  fill=white, rounded corners=1pt]
                at ($(mpref_T.south)!0.5!(mpart_ij.north)$) {$\mathbf{C}_i'$};

        \end{tikzpicture}
        \caption{DCSR multiplication (e.g., ~\Cref{lst:emul-dcsr2-rewritten}).}
        \label{fig:dcsr_kernels}
    \end{subfigure}

    \caption{Kernels for single and nested intersections. \textbf{(b)} connects two sets of kernels load balanced by the strategy in \textbf{(a)} with a prefix sum in between. Gray kernels are parallel.}
    \label{fig:kernels3-pipeline}
\end{figure}

\subsection{Parallel Assembly}
\label{sec:par-assembly}

We follow the approach of prior work~\cite{kjolstad2017taco, chou2018formats, senanayake2020scheduling, yadav2022spdistal} for computing the size of sparse output dimensions: an \textit{assembly} phase symbolically executes the sparse kernel to find allocation sizes and write offsets.
The prefix sum of the thread-local counts computes write offsets for the \textit{compute} stage, which computes and writes the result.
This is a standard parallel programming practice for writing to a sparse output without additional synchronization~\cite{BlellochTR90}.
Because the partitions computed by \sys{} are irregular and not necessarily aligned
along each dimension, a slight modification to this strategy is required for hierarchical writes (e.g., writes to the row offsets vector of a CSR matrix).
To guard against duplicate counts and writes, only the parallel worker that finishes the 
iteration over a dimension performs the write or count.
%

\textbf{Appends versus Scatters.}
This parallel-writing structure is sufficient for kernels that write to the output
in the order of iteration, logically \emph{appending} to the output~\cite{ahrens2022autoscheduling}.
%
\textit{Scattering} kernels, which randomly insert into sparse outputs are incompatible with this approach,
and efficiently supporting scatter patterns has been studied~\cite{kjolstad2019workspaces, zhang2024spworkspaces, hashemian2026sparseoutput, xiao2023survey}.
%
%
However, the community has not developed a technique for \textit{parallel} sparse scatters into 
arbitrary sparse output formats.
\sys{} instead takes the expand-sort-contract (ESC)~\cite{bell2012esc, dalton2015mergepath} approach, which does not fuse reductions when scattering into a sparse output.
ESC instead materializes (expands) the full sparse output, sorts with respect to the reduction coordinate, and 
reduces (contracts) the sorted intermediate.
The sort turns the reduction into a segmented reduction with an append-only output pattern, which can be parallelized via known techniques~\cite{cccl, baxter2016moderngpu}.
Our load-balancing technique is only applied to the materialization of the temporary, as in~\citet{dalton2015mergepath}.


%


\begin{figure}
    \centering
\input{figures/emul_csr2_compute_partition}
\end{figure}

\subsection{Parallel Compute}
\label{sec:par-compute}

\taco \textit{compute} kernels are modified similarly to assembly kernels: we adjust loop bounds to be based on stored partitions, perform writes via the offsets 
computed in the assembly phase, and guard higher-order writes to ensure only a single thread performs each write.
\Cref{lst:emul-csr2-partition} illustrates these changes for a CSR Hadamard product: only loop bounds (Lines~\ref{line:i-bounds-mod}--\ref{line:j-bounds-end}) and writes (Lines~\ref{line:write-offset}--\ref{line:pos-guard} and~\ref{line:write-guard-begin}) were modified.


\section{Evaluation}
\label{sec:evaluation}

To show that \sys{} effectively load balances across parallel processors, we provide evidence for the following claims:

\begin{enumerate}
    \item \sys{} load balances in the presence of skew.

    \item \sys{} achieves reasonable performance when compared to state-of-the-art handwritten kernels.


    \item Generality across \textit{data structures} improves performance.
    \item Load-balancing fused expressions improves performance over load-balanced binary expressions.
    
\end{enumerate}

Our results show that \sys-generated kernels achieve parity with hand-parallelized kernels from Intel MKL~\cite{mkl}, cuSPARSE~\cite{cusparse}, scheduled \taco{}~\cite{kjolstad2017taco, senanayake2020scheduling} kernels, and significantly outperform generic kernels and compound expressions from PyTorch Sparse~\cite{pytorch_sparse}.
%
\sys{} takes, on average, 7ms to generate the C++/CUDA kernel for each benchmark. 

\textbf{Methodology and Experimental Setup.}
We evaluate on a 24-core Intel i9-14900K, and a 24GB NVIDIA RTX 4090.
Benchmarks are pinned to the 8 performance cores of the CPU using \texttt{numactl}.
%
%
We used {\tt g++} 11.4.0 and CUDA 12.6, and compiled with {\tt -O3}.
Each result was collected by performing a warm-up run followed by the mean of 14 runs, where
the 2 fastest and slowest runs were excluded.
%
We evaluate on real-world tensors from SuiteSparse \cite{davis2011suitesparse} and FROSTT~\cite{frosttdataset}, in addition to synthetic tensors. Benchmarks clarify which datasets they are evaluated on.
Reported runtimes are end-to-end for all systems, including analysis phases (partitioning and precompute/symbolic phases), and output allocations, as well as the format conversions performed by PyTorch Sparse for some operations.

\textbf{Overhead of \textsc{Nacho} Partitioning.}
\sys{}'s partitioning introduces low overhead.
%
When evaluating CSR addition on the GPU across the SuiteSparse dataset, partitioning constitutes an average of $1$\% of the overall execution time (maximum $2.5$\%). In contrast, recursive partitioning for the SpGEMM kernel on the GPU constitutes an average of $7.9$\% of the total runtime (maximum $34.2$\%). Notably, this overhead includes the computation of the outer intersection required for the remapping step of our recursive partitioning algorithm, which is unavoidable even if recursive partitioning is not being used. Excluding this computation, the runtime overhead for just the two partitioning kernels in recursive partitioning for SpGEMM has an average of $2.7$\% (maximum $11.9$\%).
All presented runtimes include partitioning time.
%

\subsection{\textsc{Nacho} Load Balances in the Presence of Skew}

\begin{figure}[b]
    \centering
    \begin{subfigure}{0.48\linewidth}
        \centering
        \includegraphics[width=\linewidth, trim=0.5em 0.5em 0 1.7em, clip]{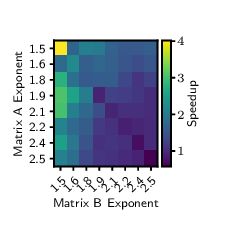}
        \caption{CPU results.}
        \label{fig:skew-csr-add-cpu}
    \end{subfigure}
    \hfill
    \begin{subfigure}{0.48\linewidth}
        \centering
        \includegraphics[width=\linewidth, trim=0.5em 0.5em 0 1.7em, clip]{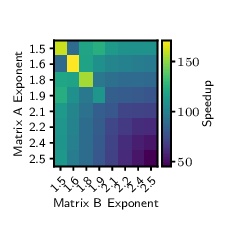}
        \caption{GPU results.}
        \label{fig:skew-csr-add-gpu}
    \end{subfigure}
    \caption{Speedup over \taco{} of CSR addition ($A + B$) on synthetic power-law matrices of different exponents of size $10000 \times 1000000$. Lower exponents indicate higher skew.} 
    \label{fig:eval-skew-csr-add}
\end{figure}

%
\Cref{fig:skew-csr-add-cpu} plots the speedup \sys{} achieves over \taco{}'s CPU-parallel CSR addition
on skewed power-law matrices, which have heavy imbalance of nonzeros between rows.
\taco{} parallelizes the outer for-loop of the addition, and uses work-stealing~\cite{blumofe1999workstealing}
to dynamically load-balance; with higher skew, \sys{} achieves better load balance by creating equally-sized
work items in the first place.
\Cref{fig:skew-csr-add-gpu} performs the same evaluation on the GPU, where \sys{} achieves significantly
higher speedups over \taco{}, which assigns each row of the addition to a GPU thread.
As the GPU has orders of magnitude more parallelism to saturate than the CPU, \sys{}'s load-balancing
is critical for performance.



\begin{figure}[b]
    \centering
    \begin{subfigure}{0.48\linewidth}
        \centering
        \includegraphics[width=\linewidth, trim=0 0.5em 0 0.5em, clip]{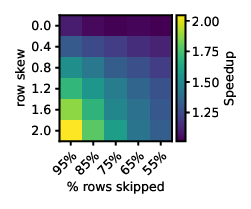}
        \caption{CPU results.}
        \label{fig:dcsr-lb-heatmap-cpu}
    \end{subfigure}
    \hfill
    \begin{subfigure}{0.48\linewidth}
        \centering
        \includegraphics[width=\linewidth, trim=0 0.5em 0 0.5em, clip]{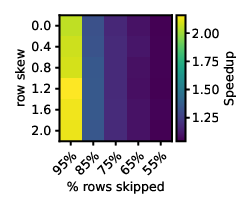}
        \caption{GPU results.}
        \label{fig:dcsr-lb-heatmap-gpu}
    \end{subfigure}
    \vspace{-0.5em}
    \caption{Speed-up of DCSR Hadamard product on synthetic $100000 \times 100000$ matrices from using recursive partitioning (\Cref{alg:recursive-search-alg}) over hierarchical partitioning (\Cref{alg:hierarchical-search-v2}).}
    \label{fig:dcsr-lb-heatmap}
\end{figure}

Next, we evaluate the impact of recursive partitioning on partition quality.
\Cref{fig:dcsr-lb-heatmap} shows the speedup of using recursive partitioning
over using only hierarchical partitioning on the element-wise multiplication
of two DCSR matrices.
When the outer intersection results in many rows being skipped, the recursive
partitioning approach yields significant speedups by constructing partitions
that better align with the true iteration costs.
On our 8-core CPU, each hierarchical partition contains many (potentially skewed heavy) rows;
when intersection iteration skips these rows, these partitions diverge from the 
true cost computed by our recursive partitioning algorithm and yield higher speedup.
This effect is mitigated on the GPU, as \sys{} computes fine-grained partitions for each thread.
Heavy rows are already split across multiple threads, causing the observed speedup to vary
only with the percent of skipped rows.

\subsection{Performance Compared to Handwritten Kernels}
\label{sec:perf-handwritten}

We show that \sys{}-generated code matches the performance of hand-written
kernels by comparing it to hand-optimized parallel kernels in PyTorch Sparse~\cite{pytorch_sparse},
which uses Intel MKL~\cite{mkl} and cuSPARSE~\cite{cusparse} when possible.
%
For operations unsupported by vendor libraries, PyTorch converts to COO and dispatches
to hand-written COO kernels.


%

\begin{figure}
    \centering
    \begin{subfigure}{\linewidth}
        \centering
        \includegraphics[width=0.9\linewidth]{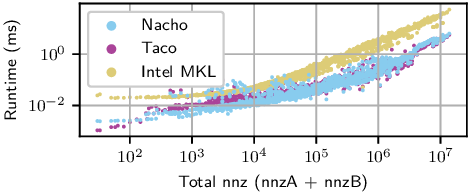}
        \caption{CPU, geo-mean speedup: $1.2\times$ (\taco) and $5.5\times$ (Intel MKL).}
        \label{fig:csr-add-cpu}
    \end{subfigure}
    \vfill
    \begin{subfigure}{\linewidth}
        \centering
         \includegraphics[width=0.9\linewidth]{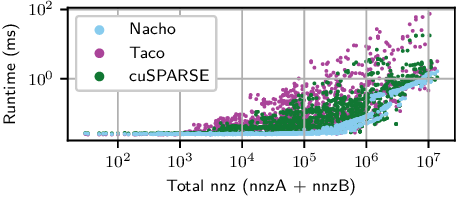}
        \caption{GPU, geo-mean speedup: $2.4\times$ (\taco) and $1.8\times$ (cuSPARSE).}
        \label{fig:csr-add-gpu}
    \end{subfigure}
    \caption{CSR addition on pairs of SuiteSparse matrices.}
    \label{fig:csr-add}
\end{figure}

We compare \sys{} with these vendor libraries for CSR addition in~\Cref{fig:csr-add-cpu,fig:csr-add-gpu}.
%
\sys{} achieves a geomean speed-up of $5.5\times$ over Intel MKL (between $0.38$--$56\times$),
and a geomean speed-up of $1.8\times$ over cuSPARSE (between $0.13$--$54\times$).
\sys{} is faster than Intel MKL on 98\% of matrices and faster than cuSPARSE on 92\% of matrices.



\begin{figure}
    \centering
    \begin{subfigure}[t]{\linewidth}
        \centering
        \includegraphics[width=0.9\linewidth]{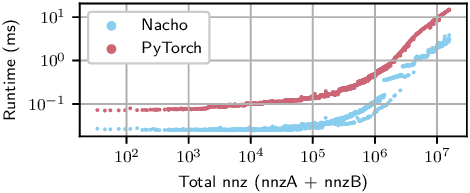}
        \caption{COO addition, geo-mean speedup: $4.1\times$.}
        \label{fig:coo-add-gpu}
    \end{subfigure}\vfill
    \begin{subfigure}[t]{\linewidth}
        \centering
        \includegraphics[width=0.9\linewidth]{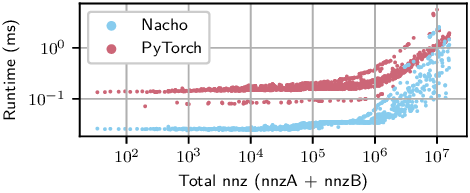}
        \caption{COO Hadamard product, geo-mean speedup: $5.3\times$.}
        \label{fig:coo-mul-gpu}
    \end{subfigure}
    \caption{GPU kernels on pairs of SuiteSparse matrices.}
    \label{fig:eval-coo-kernels}
\end{figure}

The cuSPARSE library does not support COO element-wise operations, so PyTorch's GPU backend
uses custom implementations that do not perform coiteration and are thus
asymptotically inefficient.
Our evaluation in~\Cref{fig:eval-coo-kernels} shows this:  \sys{} achieves geo-mean speedups of $4.1\times$ and $5.3\times$ for 
element-wise addition and multiplication, respectively.
%
%
%



Lastly, we compare against cuSPARSE on sparse matrix-matrix multiplication (using the default settings) on SuiteSparse matrices in~\Cref{fig:eval-spgemm} and see a geo-mean slowdown of $1.08\times$.
Based on symbols extracted from profiles, we believe this slowdown is due to
cuSPARSE's use of hash-based accumulators~\cite{nagasaka2017spgemm} as opposed to our
ESC-based reduction~\cite{bell2012esc}, which materializes the intermediate temporary before sorting and reducing.
%
%
We believe these results warrant investigation into parallelized versions of~\citet{zhang2024spworkspaces}
to improve the performance of scatter-based kernels.

\begin{figure}[t]
    \centering
    \includegraphics[width=0.9\linewidth]{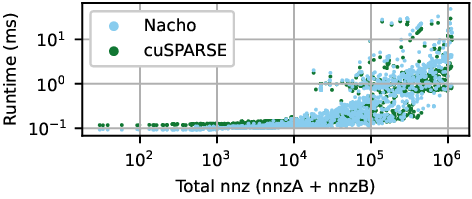}
    \caption{SpGEMM CSR $\times$ CSR, geo-mean speedup: $0.93\times$.}
    \label{fig:eval-spgemm}
\end{figure}




\subsection{Generality Across Data Structures}
\label{sec:data-structure-generality}

To demonstrate the benefit of code specialized to particular data structures, we evaluate operations on sparse data structures that existing systems do not support.
First, we evaluate third-order tensor addition on tensors from the FROSTT~\cite{frosttdataset} dataset 
that fit in GPU memory in~\Cref{tbl:addition}.
\sys{}'s COO already outperforms PyTorch Sparse's COO on both CPU and GPU: PyTorch's CPU implementation is unparallelized, and its GPU implementation is asymptotically inefficient, as it does not perform coiteration.
\sys{}'s parallel CPU kernel even surpasses PyTorch's GPU results.
Supporting CSF yields further gains: \sys{}'s CSF kernel outperforms \sys{}'s COO kernel except one case where it is close to parity, demonstrating the concrete benefit of specializing to a more compressed format.
This performance boost is likely due to reduced memory traffic in loading the CSF tensors as opposed to the COO tensors.

\begin{table}[t]
\centering
\small
\caption{Runtimes (ms) of third-order tensor addition on FROSTT~\cite{frosttdataset} tensors. PyTorch (PT) uses COO and has integer overflow errors on nell-1.}
\label{tbl:addition}

\begin{tabular}{lccc|ccc}
\toprule
& \multicolumn{3}{c|}{\textbf{CPU}} & \multicolumn{3}{c}{\textbf{GPU}} \\
\cmidrule(lr){2-4} \cmidrule(lr){5-7}
\textbf{Tensor}
& \textbf{CSF}
& \textbf{COO}
& \textbf{PT}
& \textbf{CSF}
& \textbf{COO}
& \textbf{PT} \\
\midrule
\multicolumn{1}{l|}{darpa}  & \textbf{17.78} & 47.04 & 1013.75 & \textbf{5.75} & 10.32 & 69.34 \\
\multicolumn{1}{l|}{nell-2} & \textbf{50.02} & 130.65 & 2687.37 & \textbf{10.34} & 25.15 & 331.29 \\
\multicolumn{1}{l|}{fb-m}   & 195.24 & \textbf{181.55}  & 3674.43 & \textbf{40.83} & 84.02 & OOM \\
\multicolumn{1}{l|}{nell-1} & \textbf{114.49} & 246.62 & Err & \textbf{24.93} & 63.52 & Err \\
\bottomrule
\end{tabular}
\end{table}

Next , we evaluate element-wise addition of a COO matrix with a CSR matrix in~\Cref{fig:coo-csr-add}.
\sys{} outperforms PyTorch, which converts the CSR operand to COO before performing the (non-coiterative)
COO addition.
\sys{}'s speedup comes from asymptotically efficient coiteration-based kernels and
by avoiding the unnecessary format conversion.

\begin{figure}[b]
    \centering
    \begin{subfigure}{\linewidth}
        \centering
        \includegraphics[width=0.9\linewidth]{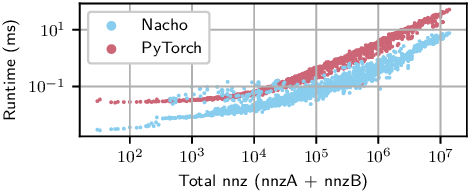}
        \caption{CPU results, geo-mean speedup $4.0\times$.}
        \label{fig:coo-csr-add-cpu}
    \end{subfigure}
    \begin{subfigure}{\linewidth}
        \centering
        \includegraphics[width=0.9\linewidth]{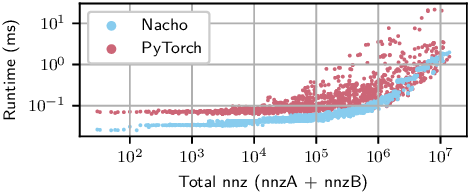}
        \caption{GPU results, geo-mean speedup $2.0\times$.}
        \label{fig:coo-csr-add-gpu}
    \end{subfigure}
    \caption{COO matrix added to CSR matrix on SuiteSparse. 
    }
    \label{fig:coo-csr-add}
\end{figure}

Finally, we evaluate CSR element-wise multiplication in~\Cref{fig:csr-mul}.
This operation is not supported by Intel MKL or cuSPARSE, so PyTorch implements it with
a conversion to COO and an asymptotically inefficient, general-purpose kernel.
\sys{} again achieves significant speedup by avoiding conversions and generating
coiterative code.

\begin{figure}
    \centering
    \begin{subfigure}{\linewidth}
        \centering
        \includegraphics[width=0.9\linewidth]{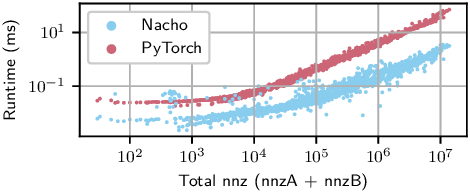}
        \caption{CPU results, geo-mean speedup $9.1\times$.}
        \label{fig:csr-mul-cpu}
    \end{subfigure}
    \begin{subfigure}{\linewidth}
        \centering
        \includegraphics[width=0.9\linewidth]{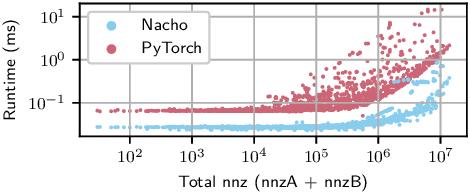}
        \caption{GPU results, geo-mean speedup $3.8\times$.}
        \label{fig:csr-mul-gpu}
    \end{subfigure}
    \caption{CSR Hadamard product on SuiteSparse matrices. 
    }
    \label{fig:csr-mul}
\end{figure}

\subsection{Load Balancing Fused Kernels}
\label{sec:expression-generality}

Because our partitioning algorithm is general across expressions, \sys{}
can extract both the constant-factor and asymptotic~\cite{kjolstad2017taco}
improvements that come from fusing sparse tensor algebra operations into 
a single kernel.
We show that \sys{} can effectively load balance these fused kernels, allowing
for additional speedup over parallelized separate kernels.


\begin{figure}
    \centering
    \begin{subfigure}[t]{\linewidth}
        \centering
        \includegraphics[width=0.9\linewidth]{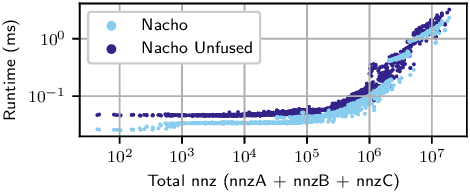}
        \caption{Three-way CSR $\oplus$, geo-mean speedup: $1.4\times$ over unfused.}
        \label{fig:eval-csr-add-fused}
    \end{subfigure}
    \begin{subfigure}[t]{\linewidth}
        \centering
        \includegraphics[width=0.9\linewidth]{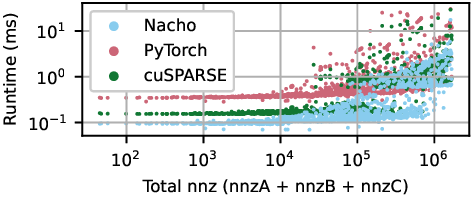}
        \caption{Sampled SpGEMM (SSSMM), geo-mean speedup: $3.2\times$ over
        PyTorch and $1.8\times$ over cuSPARSE with \sys{}'s $\odot$.}
        \label{fig:eval-sssmm}
    \end{subfigure}
    \caption{Fusion benefits across SuiteSparse matrices.}
    \label{fig:eval-fusion}
\end{figure}

\Cref{fig:eval-csr-add-fused} compares the performance of a fused three-way CSR addition
with \sys{} compared against an unfused, but parallelized, \sys{} implementation
(\Cref{fig:csr-add-gpu} demonstrates \sys{} outperforms cuSPARSE as a baseline).
The geo-mean speedup is $1.4\times$ over the unfused kernels.
Next,~\Cref{fig:eval-sssmm} evaluates fused Sampled-SpGEMM (SSSMM) \cite{milakovic2022sssmm} over 
PyTorch, which uses both a cuSPARSE SpGEMM and element-wise multiplication, showing significant speedup.
We additionally compare against an implementation that uses a cuSPARSE SpGEMM followed by a
\sys{}-generated element-wise multiplication, the best implementation of each binary operator.
Despite \sys{}'s SpGEMM being slower than cuSPARSE (\Cref{fig:eval-spgemm}), \sys{}'s fused SSSMM
kernel outperforms the unfused cuSPARSE by $1.8\times$ (geo-mean).

Lastly, we evaluate the third-order inner product~\cite{kjolstad2017taco} on FROSTT~\cite{frosttdataset}, on shifted versions of the same tensor in~\Cref{tbl:inner-product}.
PyTorch executes the element-wise multiplication and reduction as separate kernels, unlike \sys{}.
%
%
%
The relative performance of CSF and COO varies by tensor and backend.
While CSF benefits from hierarchical skipping on some tensors, it incurs higher thread divergence on other tensors.



\begin{table}[h]
\centering
\small
\caption{Runtime (ms) of third-order inner product on FROSTT~\cite{frosttdataset}. PyTorch (PT) uses COO, and has integer overflow errors on nell-1.}
\label{tbl:inner-product}

\setlength{\tabcolsep}{4pt}
\begin{tabular}{lccc|ccc}
\toprule
& \multicolumn{3}{c|}{\textbf{CPU}} & \multicolumn{3}{c}{\textbf{GPU}} \\
\cmidrule(lr){2-4} \cmidrule(lr){5-7}
\textbf{Tensor}
& \textbf{CSF}
& \textbf{COO}
& \textbf{PT}
& \textbf{CSF}
& \textbf{COO}
& \textbf{PT} \\
\midrule
\multicolumn{1}{l|}{darpa}  & \textbf{0.11} & 11.16 & 116.949 & \textbf{0.283} & 0.861 & 5.525 \\
\multicolumn{1}{l|}{nell-2} & \textbf{0.30} & 32.03 & 339.860 & \textbf{0.177} & 2.722 & 15.899 \\
\multicolumn{1}{l|}{fb-m}   & 74.99 & \textbf{49.22} & 681.892 & \textbf{8.996} & 10.749 & 20.411 \\
\multicolumn{1}{l|}{nell-1} & \textbf{13.16} & 61.39 & Err & 50.363 & \textbf{11.537} & Err \\
\bottomrule
\end{tabular}
\end{table}

\section{Related Work}

\textbf{Sparse Tensor Algebra Compilers.}
A foundational theory for sparse tensor algebra compilation was developed in the \taco
compiler~\cite{kjolstad2017taco,kjolstad2020thesis}, with extensions targeting different hardware~\cite{hsu2023sam,yadav2022spdistal,senanayake2020scheduling,hsu2025stardust}
and supporting a wider variety of tensor operations~\cite{henry2021complement, root2024burrito, liu2024spconv, sundram2024recurrences}.
These compilers expose parallelism through \emph{scheduling languages} that describe
execution strategies for nested loops~\cite{yadav2022spdistal, senanayake2020scheduling}.
However, these scheduling languages only support parallel execution of dense loops or on
tensor expressions that only contain a single sparse tensor.
Automatic scheduling frameworks have been developed that automate other aspects of scheduling, 
such as loop ordering and temporary insertion~\cite{ahrens2022autoscheduling, deeds2025galley, yan2026scorch}.
Sparse tensor compilers for specific domains, such as sparse machine learning~\cite{ye2023sparsetir},
support alternative parallelization strategies, but are also limited to tensor expressions that
have a single sparse tensor.
Compiling sparse tensor algebra to PyTorch~\cite{pytorch} has been explored to improve performance
of expressions that can utilize fixed-function units for matrix-multiplication on GPUs~\cite{won2026insum}.
WingSpan~\cite{gushin2026wingspan}, a concurrent work, extends Finch~\cite{ahrens2025finch} with support for parallelization via explicit annotations on both the data structures and the loops over sparse data structures.
Notably, WingSpan provides dependence analysis that enables sparse parallel workspaces to support scattering without intermediate materialization, but their parallelization scheme does not guarantee load balancing work.
Finally, several approaches employ inspector-executor techniques to optimize sparse tensor operations~\cite{parsy, sympiler, zhao2022spf, strout2018spf}. 
ParSy and Sympiler~\cite{parsy, sympiler} leverage compile-time analysis of tensor structures to generate parallelized or vectorized code.
The Sparse Polyhedral Framework~\cite{zhao2022spf, strout2018spf} supports parallel iterate-locate kernels~\cite{zhao2022spf}. 
However, unlike \sys, these frameworks do not support parallelizing general coiteration.

\textbf{Merge Path Partitioning.}
A subtle but important difference between partitioning strategies for merge sort~\cite{saher2012mergepath, green2012gpumergepath} and partitioning strategies for sparse tensor algebra 
is that intersections and unions require matching coordinates across operands to be owned by the same partition~\cite{odemuyiwa2023drt, dalton2015mergepath}. 
Other (non-coiterative) strategies for computing intersections and unions may be load balanced in other ways;
we focus on a load-balancing strategy specifically for coiterative merge-based unions and intersections. 
This constraint motivates our algorithm searching \textit{in coordinate space}, as opposed to position space (the space of the iterators into sparse operands).

\textbf{Parallelizing Specific Sparse Kernels.}
Load-balanced parallel algorithms for sparse matrix operations such as sparse matrix-vector multiplication,
sparse matrix addition, and sparse matrix-matrix multiplication were proposed by Dalton et al.~\cite{dalton2015mergepath}.
These algorithms partition the merge path used to compute the union and intersection of 
non-zero sparse tensor coordinates, similar to parallel merging algorithms in the sorting community~\cite{green2012gpumergepath, saher2012mergepath}.
We show a generalization of these algorithms to allow for load-balanced parallelization of 
any sparse tensor algebra expression.
Significant work has gone into optimizing sparse matrix-matrix multiplication~\cite{dalton-spgemm, liu-spgemm},
additionally focusing on parallelization of the construction of the sparse result.

\textbf{Parallel Hardware For Sparsity.}
%
%
Various specialized hardware has been proposed to accelerate sparse tensor algebra~\cite{extensor, outerspace, mapraptor, tensaurus, gamma}.
Programmable dataflow hardware for arbitrary sparse tensor
algebra expressions~\cite{hsu2023sam, koul2025onyx} extract pipeline parallelism from coiteration
and data parallelism from parallel outer loops.
We believe the partitioning strategies presented in \sys{}  could be used by this hardware to
achieve better load-balanced data-parallel execution.
\section{Conclusion}
We present a parallel partitioning algorithm that guarantees load balance for sparse tensor algebra with multiple sparse operands.
Our implementation in \sys specializes this partitioning algorithm to concrete sparse data structures and a loop order to generate efficient partitioning kernels for both CPUs and GPUs.
We match hand-parallelized sparse tensor kernels while maintaining generality across data structures and expressions.
We believe our techniques can be extended to produce partitions for load-balanced execution on supercomputers and custom sparse accelerators, and to generate kernels for other irregular domains such as database queries.

\begin{acks}
We would like to thank (in no particular order): Aart Bik, Ben Driscoll, Chris Gyurgyik, Dalia Shaaban, Devanshu Ladsaria, Evan Williams, Gautham Ravipati, James Dong, Katherine Mohr, Michael Pellauer, Olivia Hsu, Rupanshu Soi, Shiv Sundram, and Shoaib Kamil, for their helpful feedback on this draft.
This material is based upon work supported by the Defense Advanced Research Projects Agency (DARPA) under Contract HR0011-25-2-0039, as part of the Machine
Learning and Optimization-guided Compilers for Heterogeneous Architectures (MOCHA) program.
Alexander was supported by the Qualcomm Innovation Fellowship during this work.
\end{acks}

\bibliographystyle{ACM-Reference-Format}
\bibliography{references}

\appendix
\section{Artifact Appendix}

\subsection{Abstract}

The \sys compiler is provided as a C++ library, which generates CUDA/C++ code, as well as nanobind~\cite{nanobind} bindings that enable calling the generated code via Python 3. Specifying a \sys program is done via embedded C++, and then calling the compiler. We additionally provide the benchmarking scripts used to make the figures in ~\Cref{sec:evaluation}. The artifact can be evaluated on any x86-64 machine (with TBB~\cite{onetbb} support) and with an NVIDIA GPU (with CCCL~\cite{cccl} support). The benchmarks require at least 24 GB RAM for both the CPU and GPU, and at least 200 GB disk space. It is possible to restructure the evaluation to use less disk space (e.g., download a few SuiteSparse~\cite{davis2011suitesparse}/FROSTT~\cite{frosttdataset} tensors, evaluate, delete them, evaluate the next batch), but our artifact is not set up to do this.

\subsection{Artifact check-list (meta-information)}


{\small
\begin{itemize}
  \item {\bf Compilation: } We use \texttt{g++} 11.4.0 and CUDA 12.6, and compile with \texttt{-O3} to compile the generated code at C++20. \texttt{g++} 11.4.0 is also used to compile the \sys compiler. The generated CUDA is compiled for \texttt{sm\_89}; override \texttt{CMAKE\_CUDA\allowbreak\_ARCHITECTURES} for a different GPU. 

  \item {\bf Data set: } We evaluate on synthetic data for Figures~\ref{fig:eval-skew-csr-add} and ~\ref{fig:dcsr-lb-heatmap}, with the synthetic data generators in \url{benchmarks/powerlaw_heatmap_csr_add.py} and \url{benchmarks/dcsr_lb_comparison.py}, respectively; SuiteSparse~\cite{davis2011suitesparse} for Figures~\ref{fig:csr-add}, \ref{fig:eval-coo-kernels}, \ref{fig:eval-spgemm}, \ref{fig:coo-csr-add}, \ref{fig:csr-mul}, and \ref{fig:eval-sssmm}; and FROSTT~\cite{frosttdataset} for Tables~\ref{tbl:addition} and~\ref{tbl:inner-product}.

  \item {\bf Run-time environment: } Ubuntu 22.04 with the NVIDIA 560.35 driver, Python 3.12.12, and PyTorch 2.8.0+cu128. The remaining Python dependencies are pinned in \url{requirements.txt}: NumPy 2.3.3, SciPy 1.17.0, pandas 2.3.2, Matplotlib 3.10.6, networkx 3.5, nanobind 2.12.0, CMake $\geq$ 3.30, Ninja, and scikit-build-core $\geq$ 0.10. oneTBB 2022 (\texttt{libtbb.so.12}) provides the CPU runtime. Generated figures are typeset with \LaTeX{}, which needs \url{texlive-latex-base}, \url{texlive-latex-extra}, \url{cm-super}, \url{dvipng}, and \url{ghostscript}.

  \item {\bf Hardware: } We evaluate on a 24-core Intel i9-14900K (8 performance cores with SMT, and 16 efficiency cores) with 188 GB of RAM, and a 24 GB NVIDIA RTX 4090. Benchmarks are pinned to the 16 hardware threads of the 8 performance cores using \texttt{numactl --physcpubind 0-15}.

  \item {\bf Execution: } \texttt{bash run\_all\_benchmarks.sh} runs every experiment in sequence. Each script is also a standalone entry point, so individual figures can be reproduced on their own.

  \item {\bf Metrics: } We evaluate end-to-end runtime of sparse tensor algebra kernels, where the arguments are already in GPU memory for GPU benchmarks. Each result was collected by performing a warm-up run followed by the mean of 14 runs, where the 2 fastest and slowest runs were excluded. Reported runtimes are end-to-end for all systems, including analysis phases (partitioning and precompute/symbolic phases), and output allocations, as well as the format conversions performed by PyTorch Sparse for some operations.

  \item {\bf Output: } Running the artifact should reproduce Figures~\ref{fig:eval-skew-csr-add}, \ref{fig:dcsr-lb-heatmap}, \ref{fig:csr-add}, \ref{fig:eval-coo-kernels}, \ref{fig:eval-spgemm}, \ref{fig:coo-csr-add}, \ref{fig:csr-mul}, and \ref{fig:eval-sssmm}, and Tables ~\ref{tbl:addition} and~\ref{tbl:inner-product}. Figures are written to \texttt{benchmark\_results/} as EPS and PDF, alongside an \texttt{.npz} of the raw series for each; the tables are printed to standard output.


  \item {\bf How much disk space required (approximately)?: } 200 GB
  \item {\bf How much time is needed to prepare workflow (approximately)?: } About 20 minutes for the setup. Build should take less than 5 minutes. Dataset download for running the benchmarks depends on the network bandwidth. Can take 1-2 hrs for the downloading the complete 145 GB datasets.

  \item {\bf How much time is needed to complete experiments (approximately)?: } About 3 hrs.

  \item {\bf Publicly available?: } At \url{https://github.com/rootjalex/nacho}
  \item {\bf Code licenses: } MIT License.
  \item {\bf Archived: } \url{https://doi.org/10.5281/zenodo.21926844}
\end{itemize}
}

\textbf{Methodology and Experimental Setup.}

\subsection{Description}

\subsubsection{How to access}

The artifact is a single Git repository, available at \url{https://github.com/rootjalex/nacho} and archived at \url{https://doi.org/10.5281/zenodo.21926844}. It contains the \sys compiler (\texttt{src/}, \texttt{include/}), the kernel declarations evaluated in the paper (\texttt{compiler.cpp}), the runtime (\texttt{runtime/}, \texttt{python/}), the baselines we compare against (\texttt{baselines/}), and the benchmark scripts (\texttt{benchmarks/}).

\subsubsection{Hardware dependencies}
An x86-64 machine and a CUDA-capable NVIDIA GPU. The x86-64 machine should have at least 16 physical cores. The GPU experiments need a device with at least 24 GB of memory.
\subsubsection{Software dependencies}
A CUDA toolkit with \texttt{nvcc} either on \texttt{PATH} or named by \texttt{CUDACXX}, oneTBB, a C++20 compiler, Python 3.12 or newer, and \LaTeX{} for the figures. Every Python dependency except PyTorch is installed by \texttt{requirements.txt}; PyTorch is installed separately because its wheel is tied to a CUDA version, and \LaTeX{} because it is not a Python package. We tested with the versions listed under {\em Run-time environment} above.

\subsubsection{Data sets}
Two datasets are required. Locations for both are set by environment variables:

\begin{lstlisting}[basicstyle=\scriptsize\ttfamily]
export NACHO_SUITESPARSE_DIR=/path/to/suitesparse
export NACHO_FROSTT_DIR=/path/to/frostt
\end{lstlisting}

SuiteSparse~\cite{davis2011suitesparse} supplies the real matrices. \url{benchmarks/datasets/download_suitesparse_reals.sh} fetches the 1712 real-valued matrices we use.

FROSTT~\cite{frosttdataset} supplies the third-order tensors for Tables~\ref{tbl:addition} and~\ref{tbl:inner-product}: \texttt{nell-2}, \texttt{nell-1}, \texttt{1998DARPA}, and \texttt{fb-m}. \url{download_frostt.py} fetches them and \url{sort_frostt.py} deduplicates and lexicographically sorts their coordinates, which the generated kernels require. \texttt{fb-m}'s coordinates are additionally permuted to \texttt{kij} order so that its short mode can lead to more compression for CSF evaluation.

\subsection{Installation}

Clone the repository:

\begin{lstlisting}[basicstyle=\scriptsize\ttfamily]
git clone https://github.com/rootjalex/nacho nacho
cd nacho
\end{lstlisting}

Install the dependencies that are not Python packages. Ensure CUDA toolkit is installed; \LaTeX{} and \texttt{numactl} from the system package manager, and oneTBB. We use the following conda environment:

\begin{lstlisting}[basicstyle=\scriptsize\ttfamily]
conda create -n nacho_env python=3.12
conda activate nacho_env
conda install conda-forge::tbb-devel conda-forge::cmake \
   nvidia::cuda-toolkit=12.6 texlive-core ghostscript
\end{lstlisting}

If you already have CUDA toolkit 12.6 and don't wish to reinstall via conda, \texttt{nvcc} has to be findable before
anything is configured:

\begin{lstlisting}[basicstyle=\scriptsize\ttfamily]
export PATH=/usr/local/cuda-12.6/bin:$PATH
\end{lstlisting}

Then install the Python dependencies.

\begin{lstlisting}[basicstyle=\scriptsize\ttfamily]
pip install -r requirements.txt
pip install torch \
  --index-url https://download.pytorch.org/whl/cu126
\end{lstlisting}

Build the \sys compiler, generate the kernels declared in \texttt{compiler.cpp}, and compile them into the Python extension:

\begin{lstlisting}[basicstyle=\scriptsize\ttfamily]
cmake -S . -B build
cmake --build build -j16
./build/compiler
pip install --no-build-isolation -ve .
\end{lstlisting}

Confirm the installation is succesful with

\begin{lstlisting}[basicstyle=\scriptsize\ttfamily]
python tests/smoke.py
\end{lstlisting}

\subsection{Experiment workflow}
Fetch the datasets once:

\begin{lstlisting}[basicstyle=\scriptsize\ttfamily]
cd "$NACHO_SUITESPARSE_DIR"
bash <nacho>/benchmarks/datasets/\
     download_suitesparse_reals.sh
find . -name '*.tar.gz' -print0 \
  | xargs -0 -P 8 -n 1 tar -xzf
cd <nacho>
python benchmarks/datasets/download_frostt.py --all
python benchmarks/datasets/sort_frostt.py --all
\end{lstlisting}

Then run every experiment:

\begin{lstlisting}[basicstyle=\scriptsize\ttfamily]
bash run_all_benchmarks.sh
\end{lstlisting}

Each line of that script is an independent entry point, so a single figure can be reproduced without running the rest. Settings are shared across scripts live in \url{benchmarks/config.py} which include dataset paths, iteration counts, launch geometry, and output directory --- each can be overridden by an environment variable of the same name prefixed with \texttt{NACHO\_}.

\subsection{Evaluation and expected results}

Every script prints a per-input line as it runs, then a summary with the geometric-mean speedup over each baseline and the number of inputs on which the generated kernel disagreed with its reference. Correctness is checked on every input, so a run that reports zero mismatches has also validated the kernels it timed. Figures are written to \texttt{benchmark\_results/}. \Cref{tbl:artifact-mapping} maps each claim in \Cref{sec:evaluation} to the script that reproduces it and the file it writes.

\begin{table}
\centering
\scriptsize
\begin{tabular}{@{}lll@{}}
\toprule
Claim & Script (\texttt{benchmarks/}) & Output \\
\midrule
Fig.~\ref{fig:eval-skew-csr-add}  & \url{heatmap_csr_add.py} & \url{powerlaw_csr_add_cpu_10000x1000000.pdf} \\
                                  &                                   & \url{powerlaw_csr_add_cuda_10000x1000000.pdf} \\
Fig.~\ref{fig:dcsr-lb-heatmap}    & \url{heatmap_dcsr_mul.py}       & \url{dcsr_lb_comparison_cpu_heatmap.pdf} \\
                                  &                                   & \url{dcsr_lb_comparison_gpu_heatmap.pdf} \\
Fig.~\ref{fig:csr-add}            & \url{csr_add.py}                  & \url{csr_add_cpu_0-1600.pdf} \\
                                  &                                   & \url{csr_add_gpu_0-1600.pdf} \\
Fig.~\ref{fig:csr-mul}            & \url{csr_mul.py}                  & \url{csr_mul_cpu_0-1600.pdf} \\
                                  &                                   & \url{csr_mul_gpu_0-1600.pdf} \\
Fig.~\ref{fig:eval-coo-kernels}   & \url{coo_add.py}                  & \url{coo_add_gpu_0-1600.pdf} \\
                                  & \url{coo_mul.py}                  & \url{coo_mul_gpu_0-1600.pdf} \\
Fig.~\ref{fig:coo-csr-add}        & \url{coo_csr_add.py}              & \url{coo_csr_add_cpu_0-1600.pdf} \\
                                  &                                   & \url{coo_csr_add_gpu_0-1600.pdf} \\
Fig.~\ref{fig:eval-spgemm}        & \url{spgemm.py}                   & \url{spgemm_gpu_0-1300.pdf} \\
Fig.~\ref{fig:eval-csr-add-fused} & \url{csr_add_3.py}                & \url{csr_add3_gpu_0-1600} \\
Fig.~\ref{fig:eval-sssmm}         & \url{sssmm.py}                    & \url{sssmm_gpu_0-130.pdf0} \\
Tbl.~\ref{tbl:addition}           & \url{frostt_tensors_add.py}       & standard output \\
Tbl.~\ref{tbl:inner-product}      & \url{inner_prod.py}               & standard output \\
\bottomrule
\end{tabular}
\caption{Each evaluation claim, the script that reproduces it, and the basename it writes into \texttt{benchmark\_results/}.}
\label{tbl:artifact-mapping}
\end{table}






\end{document}